\PassOptionsToPackage{dvipsnames,table}{xcolor}

\documentclass[manuscript,screen,nonacm,10pt]{acmart}
\setcopyright{cc}
\setcctype{by}

\AtBeginDocument{%
}

\usepackage{multirow}
\usepackage{minted}
\usepackage[normalem]{ulem}
\usepackage[linesnumbered,ruled,vlined]{algorithm2e}
\usepackage{algpseudocode}
\usepackage[dvipsnames]{xcolor}
\usepackage[table]{xcolor}
\usepackage{adjustbox}

\newlength{\commentWidth}
\usepackage{amsmath}

\usepackage{balance}

\usepackage[utf8]{inputenc}
\usepackage{graphicx}
\usepackage{subcaption}
\usepackage{wrapfig}
\theoremstyle{definition}
\newtheorem{definition}{Definition}
\usepackage{subcaption}

\theoremstyle{example}
\newtheorem{example}{Example}

\newtheorem{theorem}{Theorem}

\newcommand{\jw}[1]{\textcolor{orange}{{\it [JW: #1]}}}

\usepackage{tikz}
\usetikzlibrary{positioning}
\usetikzlibrary{backgrounds}
\usetikzlibrary{calc}
\usetikzlibrary{fit}
\usetikzlibrary{positioning}
\usetikzlibrary{shadows}
\usetikzlibrary{shapes.geometric}
\usetikzlibrary{shapes.multipart}
\usetikzlibrary{backgrounds}
\usetikzlibrary{calc}
\usetikzlibrary{fit}
\usetikzlibrary{decorations.pathmorphing}
\usetikzlibrary{arrows.meta}
\usetikzlibrary{tikzmark}

\definecolor{mcolor}{rgb}{0.1,0.5,1}
\colorlet{lightmain}{mcolor!75}
\colorlet{lightermain}{mcolor!50}
\colorlet{lightestmain}{mcolor!25}
\colorlet{darkmain}{mcolor!75!black}
\colorlet{darkermain}{mcolor!50!black}
\colorlet{darkestmain}{mcolor!25!black}
\definecolor{MistyRose}{rgb}{1.0, 0.89, 0.88}
\definecolor{lightgreen}{RGB}{200,230,200}
\definecolor{lightorange}{RGB}{255,210,170}
\definecolor{lightred}{RGB}{255,200,200}
\definecolor{lightgray}{RGB}{230,230,230}

\newcommand*\circled[1]{\tikz[baseline=(char.base)]{
            \node[shape=circle,draw,inner sep=1pt] (char) {#1};}}

\definecolor{mygreen}{rgb}{0,0.6,0}
\colorlet{lightmygreen}{mygreen!40}

\colorlet{lightestorange}{orange!45}
\colorlet{lightestyellow}{yellow!45}

\definecolor{myorange}{HTML}{FF5522}
\definecolor{myyellow}{HTML}{FFCC11}
\definecolor{mypurple}{HTML}{AA66DD}
\definecolor{myblue}{HTML}{0077DD}
\definecolor{mypink}{HTML}{FF55BB}
\definecolor{mywhitepink}{HTML}{FFDDEE}

 \usepackage{enumitem}
 \usepackage{float}

\usepackage{bbding}
\usepackage{pifont}
\usepackage{wasysym}

\usepackage{colonequals}

\usepackage{listings}

\lstdefinestyle{motcodeStyle}{
    basicstyle=\ttfamily\small,
    keywordstyle=\color{blue}\bfseries,
    numbers=left,
    numberstyle=\tiny,
    numbersep=5pt,
    xleftmargin=1em,
    escapeinside={(*}{*)}
}

\newcommand{\NAME}{\textsc{\texttt{PPProbe}}\xspace}

\newcommand{\TOME}{\textsc{\texttt{TOME}}\xspace}
\newcommand{\MARCO}{\textsc{\texttt{MARCO}}\xspace}
\newcommand{\REMUS}{\textsc{\texttt{ReMUS}}\xspace}

\newcommand{\UNSAT}{\texttt{UNSAT}\xspace}
\newcommand{\MUS}{\texttt{MUS}\xspace}

\usepackage{caption}
\SetCommentSty{mycommentfont}

\makeatletter
\newenvironment{nineptscope}{%
  \begingroup
  \def\@typesizes{%
    \or{5}{6}%
    \or{5}{6}%
    \or{6}{7}%
    \or{7}{8}%
    \or{8}{10}%
    \or{9}{11}
    \or{10}{12}%
    \or{\@xipt}{13}%
    \or{\@xiipt}{14}%
    \or{\@xivpt}{17}%
    \or{\@xviipt}{20}%
  }%
  \normalsize
}{%
  \endgroup
}
\makeatother

\setcopyright{none}
\makeatletter
\AtEndPreamble{%
  \global\@ACM@balancefalse
  \RequirePackage{pbalance}
}
\makeatother

\begin{document}

\title{Conflict Extraction in Probabilistic Datalog Analyses}


\author{Siyu Chen}
\orcid{0009-0009-8050-9395}
\affiliation{%
  \institution{Purdue University}
  \city{West Lafayette}
  \country{USA}
}
\email{chen5216@purdue.edu}

\author{Chungha Sung}
\orcid{0009-0009-3725-6484}
\authornote{The work is not related to the author’s position in the affiliation.}
\affiliation{%
  \institution{Amazon, Inc.}
  \city{Seattle}
  \country{USA}
}
\email{chunghs@amazon.com}

\author{Xuyang Li}
\orcid{0009-0007-7978-4723}
\affiliation{%
  \institution{Purdue University}
  \city{West Lafayette}
  \country{USA}
}
\email{li5274@purdue.edu}

\author{Jingbo Wang}
\orcid{0000-0001-5877-2677}
\affiliation{%
  \institution{Purdue University}
  \city{West Lafayette}
  \country{USA}
}
\email{wang6203@purdue.edu}

\renewcommand{\shortauthors}{Chen et al.}


\begin{abstract}

Probabilistic extensions of Datalog enable static analyses such as pointer analysis, data race detection, and side-channel analysis to rank alarms by likelihood, but this added expressiveness also introduces a new challenge absent from deterministic analyses: the final output may contain alarms that are individually plausible yet mutually inconsistent, because marginal probabilities do not guarantee joint satisfiability. As a result, developers may spend effort investigating combinations of alarms that can never co-occur in any possible world. We address this problem by formalizing such inconsistencies as minimal unsatisfiable subsets (\texttt{MUSes}) and introducing \NAME, a conflict extractor specialized for probabilistic Datalog analyses. Rather than improving \MUS enumeration in general, \NAME exploits the structure of Datalog derivation graphs to guide the search toward likely conflicts and prune the search space through bottom-up \UNSAT inference. We evaluate \NAME on 70 benchmarks from power side-channel analysis, data race detection, semantic diffing, and Bayesian-network inference. The results show that \NAME achieves $2.5$--$24\times$ higher throughput than state-of-the-art \MUS enumerators, and that the conflicts it identifies yield a conservative estimate of false-positive reduction, filtering out an average of 47.7\% of mutually inconsistent alarms.

\end{abstract}

\ccsdesc[500]{Software and its engineering~Automated static analysis}
\ccsdesc[500]{Software and its engineering~Software verification}
\keywords{probabilistic Datalog, program analysis, minimal unsatisfiable subsets}

\maketitle

\section{Introduction}
\label{sec:introduction}

While classical logic programming languages such as Datalog~\cite{Ullman1988, Ceri1990} and Prolog are widely applied to domains like graph analysis~\cite{Seo2013, Seo2015, Tekle2010}, bioinformatics~\cite{Seo2018, Mungall2009}, and program analysis~\cite{Naik2006, Raghothaman2018, Zhang2017, wang2019mitigating, wang2021data, Bravenboer2009, Madsen2016, Smaragdakis2014, Whaley2005, Zhang2014, Sung2018, Kusano2016, kusano2017}, probabilistic extensions like ProbLog~\cite{De2007}, \textsc{Bingo}~\cite{Raghothaman2018}, Praline~\cite{wang2025OOPSLA} and Scallop~\cite{li2023scallop} enhance expressiveness by modeling uncertainty for ranking and risk prioritization—such as in side-channel detection, where leakage probability yields more actionable insights than binary detection~\cite{saha2023obtaining,tizpaz2019quantitative,hadvzic2024quantile}.

However, probabilistic reasoning fundamentally changes inference interpretation; unlike deterministic Datalog, which assumes a single world of facts, probabilistic languages consider exponentially many possible worlds. 
In this setting, two output facts may be valid in different worlds but mutually exclusive if no single world satisfies both--causing fundamentally incompatible outputs to appear simultaneously.
%
%
%
This lack of visibility hinders tasks like program analysis, where developers must decide which alarms to investigate and trace to root cause.
%

To systematically capture and explain these conflicts, we represent them as \emph{minimal unsatisfiable subsets \allowbreak (MUSes)}—minimal sets of output facts whose conjunction is unsatisfiable (\UNSAT) while every proper subset is satisfiable—thereby isolating the exact sources of inconsistency.
%
%

Despite extensive study, \MUS generation remains challenging because of its exponential search space.
To improve scalability, prior work has proposed both domain-agnostic~\cite{Liffiton2016FastFlexibleMUS, LiffitonSakallah2005, Liffiton2008, Bendik2018ReMUS, Bendik2016TOME, Bendik2020Unimus} and domain-aware~\cite{Van1981, Chinneck1991, Gleeson1990, Bendik2017LTL, Gasca2007} techniques.
Domain-agnostic methods iteratively manage blocking constraints to prune infeasible candidates and perform well on Boolean satisfiability benchmarks, but they do not scale well to probabilistic Datalog because they ignore the structural and statistical dependencies in Datalog derivation graphs.
Domain-aware methods use specialized knowledge to guide the search more effectively, but they are often tightly coupled to specific settings and harder to adapt to the reasoning and application constraints of probabilistic models.

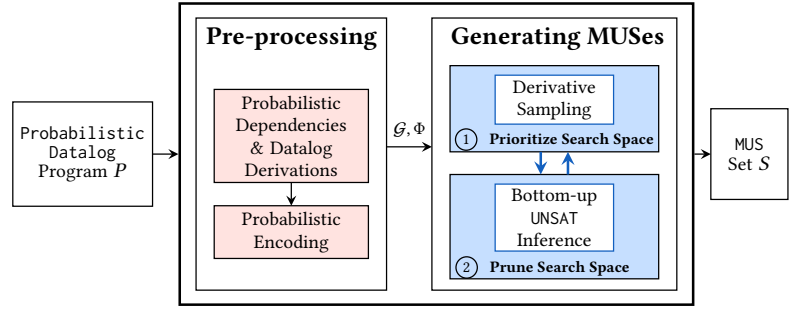
\begin{figure}[t]
\centering
\vspace{0.15in}
\scalebox{1.2}{
\begin{nineptscope}
\begin{tikzpicture}[font=\scriptsize]
\usetikzlibrary{calc,fit,backgrounds}

\tikzstyle{arrow1}=[->,>=stealth,black]
\tikzstyle{arrow2}=[thick,->,>=stealth,darkmain]

\tikzstyle{inputRec}=[
rectangle, draw,
minimum width=1.45cm, minimum height=1.15cm,
inner sep=2pt, outer sep=0pt,
align=center, text width=1.4cm]

\tikzstyle{outputRec}=[
rectangle, draw,
minimum width=0.85cm, minimum height=1.0cm,
inner sep=2pt, outer sep=0pt,
align=center, text width=0.62cm]

\tikzstyle{innerRec}=[
rectangle, draw=black,
minimum width=1.70cm, minimum height=0.44cm,
inner sep=2pt, outer sep=0pt,
align=center, text width=1.52cm]

\node at (-3.5, 0) [inputRec] (M1)
{\texttt{Probabilistic}\\[-1pt]\texttt{Datalog}\\[-1pt]Program $P$};

\node at (3.85, 0) [outputRec] (M3)
{\texttt{MUS}\\[-1pt]Set $S$};

\node at (-1.20, 0) [rectangle, draw,
minimum width=2.1cm, minimum height=3cm,
inner sep=0pt, outer sep=0pt] (D2) {};

\node at (1.70, 0) [rectangle, draw,
minimum width=2.7cm, minimum height=3cm,
inner sep=0pt, outer sep=0pt] (D3) {};

\begin{scope}[on background layer]
\node[rectangle, draw, thick, fit=(D2)(D3),
inner xsep=0.18cm, inner ysep=0.16cm] (B1) {};
\end{scope}

\node[anchor=west] at ($(D2.north west)+(0,-0.23)$) {\small \textbf{Pre-processing}};

\node at (-1.20, 0.2) [innerRec, fill=MistyRose] (A1)
{\scriptsize	 Probabilistic \\ Dependencies \\ \& Datalog Derivations};

\node at (-1.20, -0.86) [innerRec,  fill=MistyRose] (A12)
{\scriptsize	 Probabilistic Encoding};

\draw[arrow1] (A1.south) -- (A12.north);

\node[anchor=west] at ($(D3.north west)+(0.1,-0.23)$) {\small \textbf{Generating MUSes}};


\node at (1.70, 0.5) [rectangle, draw=black,fill=lightestmain,
minimum width=2.3cm, minimum height=0.95cm,
inner sep=0pt, outer sep=0pt] (P1) {};

\node at (1.70, 0.6) [rectangle, draw=darkmain, fill=white,
minimum width=1.2cm, minimum height=0.44cm,
text width=1.2cm, inner sep=1.5pt, outer sep=0pt,
align=center] (G1)
{ Derivative Sampling};

\node at (1.70, 0.16) {\tiny \circled{1} $~$  \textbf{Prioritize Search Space}};

\node at (1.70, -0.8) [rectangle, draw=black,fill=lightestmain,
minimum width=2.3cm, minimum height=1.15cm,
inner sep=0pt, outer sep=0pt] (P2) {};

\node at (1.70, -0.7) [rectangle, draw=darkmain, fill=white,
minimum width=1.2cm, minimum height=0.44cm,
text width=1.2cm, inner sep=1.5pt, outer sep=0pt,
align=center] (G2)
{Bottom-up\\ \texttt{UNSAT} \\ Inference};

\node at (1.57, -1.25) {\tiny  \circled{2} $~$   \textbf{Prune Search Space}};

\draw[arrow1] ($(D2.east)+(0,0.08)$) --
node[midway, above] {\tiny $\mathcal{G},\Phi$}
($(D3.west)+(0,0.08)$);

\draw[arrow2] ($(P1.south) + (-0.15,0)$) -- ($(P2.north) + (-0.15,0)$);
\draw[arrow2] ($(P2.north) + (0.15,0)$) -- ($(P1.south) + (0.15,0)$);

\draw[arrow1] (M1.east) -- (B1.west);
\draw[arrow1] (B1.east) -- (M3.west);

\end{tikzpicture}
\end{nineptscope}
}
\caption{\NAME -- Overview of generating MUSes from the probabilistic Datalog program $P$.  }
\vspace{-0.1in}
\Description{An overview diagram.}
\label{fig:sys}
\end{figure}

To address these gaps, we propose \NAME\footnote{Abbreviation of \emph{Probabilistic Paradox Probe}.}, a \emph{sound} and \emph{efficient} \MUS generator for probabilistic Datalog programs. 
By \emph{sound}, we mean that every \MUS produced by \NAME corresponds to a genuine conflict: the identified facts are mutually exclusive because their derivations rely on fundamentally incompatible conditions, and therefore cannot hold simultaneously in any valid interpretation. 
Figure~\ref{fig:sys} provides an overview of the approach. Given a probabilistic Datalog program $P$, \NAME first invokes a standard solver to preprocess $P$, producing the queried output facts and the associated derivation graph $\mathcal{G}$. 
These artifacts are then used to construct a constraint set $\Phi$ that encodes the probabilistic semantics of $P$, over which \NAME identifies a collection of \texttt{MUSes}, each corresponding to a minimal subset of output facts whose conjunction is unsatisfiable under $\Phi$.
%

%
%
%

To address the exponential cost of checking candidate subsets for unsatisfiability and minimality, \NAME employs two complementary techniques.
The first, \emph{derivation-aware sampling}, prioritizes the search by exploiting negative structural and statistical dependencies in the derivation graph $\mathcal{G}$.
By favoring relations that depend on shared inputs in opposite ways, the sampler steers candidate generation toward high-conflict regions that are more likely to be \UNSAT.
In addition, \NAME performs static analysis to partition independent output facts into disjoint sets and restricts sampling within each set, further increasing the chance of discovering unsatisfiable candidates.



Our second technique,  \emph{bottom-up \MUS inference}, aggressively prunes the search space by deriving new conflicts from previously detected ones. 
Through a process of logical replacement, \NAME traverses the derivation graph to systematically replace facts in a known \MUS with their derivation ancestors. 
This mechanism allows the framework to synthesize additional \texttt{MUSes} without the overhead of SAT solver invocations, substantially reducing the total number of calls and accelerating the enumeration process.

We implemented these techniques in \NAME and evaluated them on \textbf{70} benchmarks spanning side-channel (\textbf{SC}) vulnerability detection~\cite{wang2025OOPSLA, zhang2018scinfer}, data race (\textbf{DR}) analysis~\cite{Raghothaman2018,li2025combining}, semantic diffing (\textbf{SD})~\cite{Sung2018}, and Bayesian network inference~\cite{bench-munin}.
We compare \NAME against three state-of-the-art \MUS enumerators: \MARCO~\cite{Liffiton2016FastFlexibleMUS}, \allowbreak \REMUS~\cite{Bendik2018ReMUS}, and \TOME~\cite{Bendik2016TOME}.
%
%
Under a 30-minute time budget, \NAME discovers substantially more \texttt{MUSes} than these baselines, including average improvements of \textbf{6.6}$\times$ and \textbf{5.9}$\times$ over \MARCO, and \textbf{18}$\times$ and \textbf{65}$\times$ over \REMUS, on \textbf{SC} and \textbf{DR}, respectively.
%
%
%
To demonstrate practical utility, we further apply \NAME to three program analysis tasks (\textbf{SC}, \textbf{DR}, and \textbf{SD}) and use the detected \texttt{MUSes} to filter logically inconsistent critical queries.
This \MUS-based filtering significantly reduces the diagnostic search space, with average reductions of \textbf{61\%} in \textbf{DR} and \textbf{69\%} in \textbf{SC}.

\noindent\textbf{Contributions.}
In summary, this paper makes the following contributions:

\begin{itemize}[leftmargin=*]
    \item We propose a \emph{sound} and \emph{efficient} \MUS generator, \NAME, tailored for extracting conflicts from probabilistic Datalog programs.
    
    
    \item We develop a derivation-aware sampling technique that utilizes negative structural and statistical dependencies in the derivation graph to prioritize candidate generation.
    
    \item We present a bottom-up \MUS inference technique that proactively derives additional \texttt{MUSes} via logical replacement, pruning the exponential search space without additional solver overhead.
    \item We demonstrate the effectiveness of our approach on \textbf{70} real-world benchmarks, achieving average throughput improvements of up to \textbf{65$\times$} over state-of-the-art tools, while indicating a potential \textbf{47.7\%} average reduction in false alarms.
    
\end{itemize}

The remainder of this paper is organized as follows. Section~\ref{sec:motivation} motivates our work with illustrative examples. Section~\ref{sec:prelim} introduces the necessary preliminaries and background. Section~\ref{sec:overview} provides an overview of our approach. Sections~\ref{sec:derv_sample} and~\ref{sec:unsatInference} describe our techniques for prioritizing and pruning the search space. We present experimental results in Section~\ref{sec:eval}, discuss related work in Section~\ref{sec:related}, and conclude in Section~\ref{sec:conclusion}.

\section{Motivation}
\label{sec:motivation}

We illustrate the challenges of detecting logical conflicts in probabilistic Datalog using power side-channel analysis—a foundational approach for evaluating software security~\cite{Raghothaman2018, wang2025OOPSLA, heo2019continuously, li2025combining}.
By adapting a simplified side-channel detection framework~\cite{zhang2018scinfer}, we demonstrate how probabilistic reasoning introduces subtle, hidden conflicts that are easily overlooked in large-scale analyses.

%
%
%

\begin{figure}[htbp]
\vspace{-0.1in}
\centering

\definecolor{commentgray}{rgb}{0.5,0.5,0.5}
\definecolor{keywordblue}{rgb}{0.1,0.1,0.9}
\lstset{
    language=Prolog,
    basicstyle=\ttfamily\footnotesize,
    breaklines=true,
    captionpos=b,
    commentstyle=\color{commentgray},
    emph={inRand, inKey, and, xor, noIntersect, leak, rand, sid, query},
    emphstyle=\color{keywordblue},
    numbers=left,
    numberstyle=\tiny\color{commentgray},
    stepnumber=1,
    numbersep=10pt,
    frame=none,
    mathescape=true, 
    xleftmargin=12pt
}
\begin{lstlisting}
% Datalog facts (input): user-provided input type annotations
1.0::inRand(r1). 1.0::inRand(r2). 1.0::inRand(r3). 1.0::inKey(k). 

% Datalog facts (input): program semantics
1.0::and(c5,c3,r3). 1.0::and(c2,k,c1). 1.0::and(c3,c2,r2). 
1.0::and(c4,c1,c2). 1.0::xor(c1,k,r1). 1.0::noIntersect(c2,r2). 

% Datalog rules for inferring output types (rand, sid, key)
0.9::leak(OP1) :- inKey(OP1). % Rule 1
0.8::rand(OP1) :- inRand(OP1). % Rule 2
1.0::rand(O) :- inKey(OP1), inRand(OP2), xor(O,OP1,OP2). % Rule 3
0.7::leak(O) :- leak(OP1), rand(OP2), and(O,OP1,OP2). % Rule 4
0.8::sid(O)  :- $\neg$leak(OP1), rand(OP2), and(O,OP1,OP2), noIntersect(OP1,OP2). % Rule 5

% Datalog queries
query(leak(_)). % leak(c2)=0.63, leak(c3)=0.35, ...
query(sid(_)).  % sid(c3) =0.064
\end{lstlisting}
\vspace{-0.1in}
\caption{Side-channel analysis in probabilistic Datalog.}
\label{fig:problog_motivate}
\vspace{-0.1in}
\end{figure}

\textbf{\emph{Example Program and Encoding.}}
Figure~\ref{fig:motivateProgram} (left) depicts a code fragment, \texttt{compute()}, which utilizes a secret input \texttt{k} and three random variables \texttt{r1}--\texttt{r3} to produce a masked output through intermediate values \texttt{c1}--\texttt{c5}.
Despite the masking, the implementation contains inherent side-channel risks. 
To evaluate these, we represent the program logic in a probabilistic Datalog model (Fig~\ref{fig:problog_motivate}) that tracks information flow through three relations: \texttt{leak(var)} indicates a potential information leak, \texttt{rand(var)} denotes a variable masked by pure randomness, and \texttt{sid(var)} indicates statistical independence from the secret.

The encoding has three parts.
First, the \textit{input facts} (lines 2--6) define variable types and statement semantics, including type annotations such as \texttt{inKey(k)} and instruction semantics such as \texttt{xor(c1,k,r1)}, which corresponds to line 1 of the original program in Figure~\ref{fig:motivateProgram} (left).
All input facts are prefixed with \texttt{1.0::}, indicating deterministic program semantics.
Second, the \textit{inference rules} (lines 9--13) propagate these properties.
Each rule is assigned a probability based on empirical domain knowledge about data dependencies~\cite{wang2025OOPSLA,Raghothaman2018}.
For example, Rule 4 assigns a 70\% probability that an \texttt{AND} operation on a leaky variable (\texttt{OP1}) and a random variable (\texttt{OP2}) also leaks secret information.
Finally, the \textit{queries} (lines 16--17) identify variables that are either leaking or independent.

To produce these results, a solver repeatedly applies these rules to the input facts until reaching a fixed point. 
Throughout this process, the system constructs a \textit{derivation graph} to track the provenance of each fact and the computation of associated probabilities.

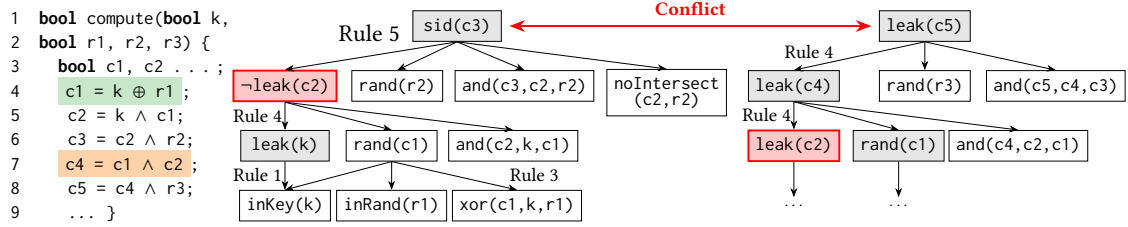
\begin{figure*}[t]
\centering
\scalebox{0.9}{
\begin{nineptscope}
  \begin{tikzpicture}[
    x=0.55cm, y=1.1cm,
    box/.style={draw, minimum width=1.3cm, minimum height=0.45cm, align=center, fill=white, font=\ttfamily\scriptsize},
    input/.style={box},
    derived/.style={box, fill=lightgray},
    conflict/.style={box, fill=lightred, draw=red, thick},
    rule_label/.style={inner sep=1pt, font=\scriptsize, anchor=north east},
    arrow/.style={-{Stealth[scale=0.8]}, rounded corners, font=\tiny},
    highlight_green/.style={fill=lightgreen},
    highlight_orange/.style={fill=lightorange}
]

\node[align=left, font=\ttfamily\small] (code) at (-3, 1.3) {
    1 \ \textbf{bool} compute(\textbf{bool} k, \\
    2 \ \textbf{bool} r1, r2, r3) \{ \\
    3 \ \ \ \textbf{bool} c1, c2 \dots; \\
    4 \ \ \ \colorbox{lightgreen}{c1 = k $\oplus$ r1}; \\
    5 \ \ \ \ c2 = k $\wedge$ c1; \\
    6 \ \ \ \ c3 = c2 $\wedge$ r2; \\
    7 \ \ \ \colorbox{lightorange}{c4 = c1 $\wedge$ c2}; \\
    8 \ \ \ \ c5 = c4 $\wedge$ r3; \\
    9 \ \ \ \ ... \}
    
};

\node[derived] (sid) at (6.0, 2.5) {\small sid(c3)};

\node[derived] (leak5) at (18.5, 2.5) {\small leak(c5)};
\node[rule_label, xshift=-13mm] at (leak5.south) {\small Rule 4};

\node[rule_label, xshift=-8mm, yshift=3mm] at (sid.south) {\large Rule 5};

\node[conflict] (not_leak2) at (1.4, 1.7) {\small $\neg$leak(c2)};
\node[input] (rand_r2) at (4.4, 1.7) {\small rand(r2)};
\node[input] (and3) at (7.8, 1.7) {\small and(c3,c2,r2)};
\node[input] (nointer) at (11.6, 1.61) {\small noIntersect\\ \small(c2,r2)};

\node[rule_label, xshift=0mm, yshift=-0.5mm] at (not_leak2.south) {\small Rule 4};

\node[derived] (leak4) at (15, 1.7) {\small leak(c4)};

\node[rule_label, xshift=0mm, yshift=-0.5mm] at (leak4.south) {\small Rule 4};
\node[input] (rand_r3) at (18.5, 1.7) {\small rand(r3)};
\node[input] (and5) at (22, 1.7) {\small and(c5,c4,c3)};

\node[derived] (leak_k) at (1.4, 0.9) {\small leak(k)};
\node[input] (rand_c1) at (4.25, 0.9) {\small rand(c1)};
\node[input] (and2) at (7.5, 0.9) {\small and(c2,k,c1)};

\node[rule_label, xshift=0mm, yshift=-0.5mm,] at (leak_k.south) {\small Rule 1};
\node[rule_label, xshift=25mm, yshift=-0.5mm] at (rand_c1.south) {\small Rule 3};

\node[conflict] (leak2) at (15, 0.9) {\small leak(c2)};
\node[derived] (rand_c1_2) at (17.8, 0.9) {\small rand(c1)};
\node[input] (and6) at (21, 0.9) {\small and(c4,c2,c1)};

\node[input] (inkey) at (1.4, 0.1) {\small inKey(k)};
\node[input] (inrand) at (4.25, 0.1) {\small inRand(r1)};
\node[input] (xor1) at (7.6, 0.1) {\small xor(c1,k,r1)};

\node[font=\ttfamily\scriptsize] (dot1) at (15, 0.1) {...};
\node[font=\ttfamily\scriptsize] (dot2) at (17.8, 0.1) {...};

\draw[arrow] (sid.south) -- (nointer.north);
\draw[arrow] (sid.south) -- (and3.north);
\draw[arrow] (sid.south) -- (rand_r2.north);
\draw[arrow] (sid.south) -- (not_leak2.north);
\draw[arrow] (not_leak2.south) -- (leak_k.north);
\draw[arrow] (not_leak2.south) -- (and2.north);
\draw[arrow] (not_leak2.south) -- (rand_c1.north);
\draw[arrow] (leak_k) -- (inkey);
\draw[arrow] (rand_c1.south) -- (xor1.north);
\draw[arrow] (rand_c1.south) -- (inrand.north);
\draw[arrow] (rand_c1.south) -- (inkey.north);

\draw[arrow] (leak5.south) -- (leak4.north);
\draw[arrow] (leak5.south) -- (rand_r3.north);
\draw[arrow] (leak5.south) -- (and5.north);

\draw[arrow] (leak4.south) -- (leak2.north);
\draw[arrow] (leak4.south) -- (rand_c1_2.north);
\draw[arrow] (leak4.south) -- (and6.north);

\draw[arrow] (leak2.south) -- (dot1.north);
\draw[arrow] (rand_c1_2.south) -- (dot2.north);

\draw[red, thick, Stealth-Stealth] ($(sid.east) + (0.2, 0)$) -- node[above=1pt, font=\bfseries\small] {Conflict} ($(leak5.west) + (-0.2, 0)$);


\end{tikzpicture}
\end{nineptscope}
}
\caption{
Left: an example program adapted from~\cite{zhang2018scinfer}.
Right: its derivation graph, where white nodes are input facts and gray nodes are derived facts.
The red highlights show a conflict: one branch requires $\neg\mathsf{leak(c2)}$, while the other requires $\mathsf{leak(c2)}$, so $\mathsf{sid(c3)}$ and $\mathsf{leak(c5)}$ cannot hold together.
Rule numbers refer to Fig.~\ref{fig:problog_motivate}.
}
\Description{Motivating example where left is an example program adapted from~\cite{zhang2018scinfer}.}
\label{fig:motivateProgram}
\vspace{-1em}
\end{figure*}

\textbf{\emph{Derivation Graph and Conflicts.}} Figure~\ref{fig:motivateProgram} (right) shows the derivation graphs for \texttt{sid(c3)} and \texttt{leak(c5)}.
White nodes denote input facts, while gray and red nodes denote intermediate or output relations.
These two outputs are logically inconsistent: \texttt{sid(c3)} can be derived only if \texttt{leak(c2)} is false (Rule 5), whereas \texttt{leak(c5)} requires \texttt{leak(c2)} to be true (Rule 4).

Probabilistic Datalog, however, uses possible-worlds semantics where uncertain facts like \texttt{leak(c2)} are assigned marginal probabilities instead of binary truth values. 
Because \texttt{leak(c2)} may be true in some possible worlds and false in others, both \texttt{sid(c3)} and \texttt{leak(c5)} can appear simultaneously in query results. 
Standard solvers typically report these probabilities independently without signaling logical incompatibility, which can be misleading in security-critical contexts where mutually exclusive results represent fundamentally different program behaviors.

\textbf{\emph{Impact and Summary.}} Without conflict awareness, developers must inspect many spurious alarms, often spending effort on combinations that cannot occur in the same execution, such as \texttt{sid(c3)} and \texttt{leak(c5)} being true together.
This problem becomes especially severe in large-scale quantitative analyses, where the output volume makes manual conflict detection impractical.
Although probabilistic Datalog is effective for modeling uncertainty~\cite{wang2025OOPSLA,Raghothaman2018, heo2019continuously, li2025combining}, it does not explicitly expose internal logical inconsistencies, which can lead to misleading results.
Our work addresses this gap with a sound and efficient \MUS generator that automatically extracts such conflicts.
By exposing these hidden inconsistencies, \NAME helps developers focus on plausible vulnerabilities and improves the precision and usability of probabilistic program analyses.

\section{Preliminaries}
\label{sec:prelim}

We introduce the key concepts used throughout this paper.  
We first define probabilistic Datalog programs and their grounded solutions, then describe how they are represented using derivation graphs, and finally formulate conflicts as \texttt{MUSes}.

\subsection{Probabilistic Datalog Programs}
\label{sec:language}
A probabilistic Datalog program consists of a set of rules $R$, where each rule is a Horn clause annotated with a probability:
$
p :: h(\vec{x}) \colonminus \odot b_1(\vec{y_1}), \ldots, \odot b_n(\vec{y_n}),
$
where $h(\vec{x})$ is the \emph{head}, each $b_i(\vec{y_i})$ is a \emph{body predicate}, and $\odot$ denotes an optional negation operator.  
The probability $p \in [0,1]$ specifies how likely the head holds when the body is satisfied.
A probabilistic Datalog program also assumes stratified negations in order to compute the least fix-point within limited time, meaning that recursive negations in rules are not allowed.

When the body of a rule is empty, it becomes an \emph{input fact declaration}, e.g., \texttt{1.0::inRand(r1)} in Figure~\ref{fig:problog_motivate}.  
A program also contains one or more \emph{queries}, which specify the output relations of interest.  
The solution to a program is a set of \emph{grounded output facts}, obtained by repeatedly applying rules to input facts until reaching a fixed point over the Herbrand universe.

%

For example, in Figure~\ref{fig:motivateProgram}, all \texttt{leak(OP1)} relations are grounded by substituting \texttt{OP1} with concrete variable names (e.g., \texttt{c2,c5}), shown as gray boxes in the figure.

\paragraph{\textbf{Derivation Graph}}
The grounded solution of a probabilistic Datalog program can be represented as a \emph{derivation graph}, which captures how each output fact is derived from input facts.

\begin{definition}[Derivation Graph]
\label{def:derivation-graph}
A derivation graph for a grounded probabilistic Datalog solution is a hypergraph $G=(V,E)$ where:
\begin{itemize}[leftmargin=*]
    \item $V$ is the set of nodes, each representing a grounded relation (fact). 
    \item $E$ is the set of directed hyperedges, where each hyperedge is a tuple $(p,h,B^+,B^-,r)$ representing a grounded rule $r$ with probability $p$, head $h$, positive body $B^+$, and negative body $B^-$. Each hyperedge represents a fully grounded rule, where every predicate (e.g., $h(\vec{x})$) has concrete terms.  

\end{itemize}
\end{definition}

The source of a hyperedge is its head $h$, and the targets are the union of body predicates $B^+ \cup B^-$.
 $B^+$ denote the set of non-negated body predicates and $B^-$ the set of negated body predicates.

\begin{example}
    In the derivation graph of Figure~\ref{fig:motivateProgram}, the hyperedge for inferring \texttt{sid(c3)} has the head $h$ = \texttt{sid(c3)}, $B^+=\{$\texttt{rand(r2)}, \texttt{and(c3,c2,r2)}, \texttt{noIntersect(c2,r2)}$\}$, $B^-=\{$\texttt{leak(c2)}$\}$, $p=0.8$ and rule $r$ is noted as Rule 5.
\end{example}

\emph{Derivation Set of a Node.}  
For a node $v \in V$, we define its \emph{derivation set} as:
$
\texttt{Derv}(v) = \{e \mid e\in E, ~e=(p,v,B^+,B^-,r)\}
$,
consisting of hyperedges.

\subsection{Boolean Encoding of the Grounded Solution}

To reason about satisfiability, we encode the derivation graph as Boolean constraints.  
This enables \texttt{SAT} and \UNSAT reasoning over sets of grounded facts.

\begin{definition}[Mapping Rules to Boolean Variables]
\label{def:boolean-vars}
Given a derivation graph $G=(V,E)$, each grounded rule $r$ is mapped by $S$ to a Boolean variable $b_r$ that follows a Bernoulli distribution, where $P(b_r = \text{true}) = p$ and $p$ is the probability associated with $r$.
\end{definition}

\begin{definition}[Mapping Relations to Boolean Variables]
\label{def:gamma}
Each grounded relation (node $v \in V$) is mapped by $\Gamma$ to a Boolean variable $b_v$, indicating whether the fact $v$ holds.
\end{definition}

\begin{definition}[Boolean Interpretation]
\label{def:bool-interpret}
Given a derivation graph $G=(V,E)$, a \emph{Boolean interpretation} is a tuple $T(G) = (\Gamma, S, \Phi)$ where:
\begin{align*}
\Gamma &= \{v \mapsto b_v \mid v \in V, b_v \in \mathbb{B}\},\\
S &= \{r \mapsto b_r \mid (p,h,B^+,B^-,r) \in E, b_r \in \mathbb{B}\},\\
\Phi &= \Big\{c \mid v \in V,\ c \triangleq \Gamma[v] = D_v,\\
&\hspace{1em} D_v \triangleq \bigvee_{\substack{(p,v,B^+,B^-,r)\\ \in\, \texttt{Derv}(v)}}
\Big(
S[r] \land \bigwedge_{b^+ \in B^+} \Gamma[b^+]
\land \bigwedge_{b^- \in B^-} \neg \Gamma[b^-]
\Big)
\Big\}.
\end{align*}
\end{definition}

The Boolean encoding $\Phi$ follows the possible-world semantics of
probabilistic Datalog.
For each grounded probabilistic rule instance $r$, the map $S$
introduces a fresh Boolean event variable $S[r]$ indicating whether that rule
instance is enabled in a possible world.
Thus, a derived fact $v$ holds only if at least one of its grounded derivations
is enabled: the corresponding rule event $S[r]$ must be true, all positive body
facts must hold, and all negated body facts must fail.
This differs from deterministic Datalog, where a satisfied rule body directly
entails the head without an additional rule-event variable.
In this way, the rule-event variables give $\Phi$ an abstraction of the
possible-world structure that does not involve the concrete numeric probability
weights.

This abstraction is sufficient for conflict extraction.
The probability annotation $p$ in a rule of the form $p::r$
(Section~\ref{sec:language}) is needed only for probabilistic inference: it
specifies the Bernoulli weight of the rule-event variable $S[r]$ when
computing marginal query probabilities.
Conflict extraction, in contrast, asks a purely Boolean question over $\Phi$:
whether some assignment to the rule-event and fact variables makes a set of
output facts jointly true.
Since the numeric weights affect only how likely a satisfiable world is, not
whether one exists, conflict extraction loses no information by operating on
the abstraction alone.



\subsection{Conflicts and MUS Formulation}
\label{sec:conflict-encoding}

We now formalize conflicts among output facts using \texttt{SAT} reasoning over the Boolean interpretation of a probabilistic Datalog derivation graph.

\begin{definition}[Joint Satisfiability]
\label{def:joint-sat}
Given a derivation graph $G=(V,E)$ and its Boolean interpretation $T=(\Gamma,S,\Phi)$, a set of nodes $V' \subseteq V$ is \emph{jointly satisfiable} if the constraint
\[
\Psi(V') = \Big(\bigwedge_{\phi \in \Phi} \phi\Big) \land \Big(\bigwedge_{v \in V'} \Gamma[v]\Big)
\]
is \texttt{SAT}.  
Otherwise, $V'$ is \emph{jointly unsatisfiable}.

Definition~\ref{def:joint-sat} checks whether the output facts in
$V'$ can co-occur in at least one possible world of the probabilistic
Datalog program. The constraint $\Phi$ encodes the logical structure of
those possible worlds, including the Boolean event variables introduced
for probabilistic rule instances. Hence, joint unsatisfiability captures
a genuine semantic conflict: no assignment to the probabilistic rule
events can make all facts in $V'$ true simultaneously. Importantly, this
notion depends on the structure of the probabilistic derivations, but
not on the concrete numeric probabilities assigned to the rule events.
The numeric weights affect how likely a satisfiable set of facts is, but
they do not affect whether such a set is logically realizable.



\end{definition}

\begin{example}
\label{example:correctness}
In the derivation graph of Figure~\ref{fig:motivateProgram}, we use
$r_5$ and $r_4$ to denote Rule~5 and Rule~4, respectively.
The set $\textcolor{darkmain}{\Phi}$ includes the following constraints:
\[
\begin{array}{rcl}
\textcolor{darkermain}{\phi_1}: &
\mathsf{sid(c3)} &= \begin{aligned}[t]
  &S(r_5)
  \wedge \underline{\neg \mathsf{leak(c2)}}
  \wedge \mathsf{rand(r2)} 
  \wedge \mathsf{and(c3,c2,r2)} \\
  &\wedge \mathsf{noIntersect(c2,r2)}
\end{aligned}, \\[2pt]

\textcolor{darkermain}{\phi_2}: &
\mathsf{leak(c5)} 
&= S(r_4)
\wedge \mathsf{leak(c4)}
\wedge \mathsf{rand(r3)}
\wedge \mathsf{and(c5,c4,c3)}, \\[2pt]

\textcolor{darkermain}{\phi_3}: &
\mathsf{leak(c4)} 
&= S(r_4')
\wedge \underline{\mathsf{leak(c2)}}
\wedge \mathsf{rand(c1)}
\wedge \mathsf{and(c4,c2,c1)} .
\end{array}
\]

For $V' = \{\mathsf{sid(c3)}, \mathsf{leak(c5)}\}$, the constraint $\Psi(V')$ is encoded as
\[
\textcolor{darkermain}{\phi_1} \wedge \textcolor{darkermain}{\phi_2} \wedge \textcolor{darkermain}{\phi_3}
\wedge \mathsf{sid(c3)} \wedge \mathsf{leak(c5)}.
\]
For brevity, we omit the remaining constraints in $\Phi$.
As shown above, $\Psi(\{\mathsf{sid(c3)},\allowbreak \mathsf{leak(c5)}\})$ is \texttt{UNSAT}, since both
\underline{$\mathsf{leak(c2)}$}and
\underline{$\neg \mathsf{leak(c2)}$}
appear in the conjunction.
Hence, $\{\mathsf{sid(c3)}, \mathsf{leak(c5)}\}$ is jointly \UNSAT.

\end{example}

\begin{definition}[Minimal Unsatisfiable Subset]
\label{def:mus-bool}
A set of nodes $V' \subseteq V$ is a \emph{Minimal Unsatisfiable Subset (\MUS)} if:
(1) $\Psi(V')$ is UNSAT, and
(2) for every strict subset $V'' \subsetneq V'$, $\Psi(V'')$ is SAT.

\end{definition}

The set $\{\mathsf{sid(c3)}, \mathsf{leak(c5)}\}$ is also a \emph{minimal unsatisfiable subset},
since each of its strict subsets, $\{\mathsf{sid(c3)}\}$ and $\{\mathsf{leak(c5)}\}$, is \texttt{SAT}.


An MUS represents the smallest group of grounded facts that cannot all be true simultaneously.  
Identifying MUSes allows us to pinpoint minimal conflicts within the program’s probabilistic reasoning.
%
Given a derivation graph $G=(V,E)$, our goal is to enumerate all MUSes $V' \subseteq V$ that satisfy Definition~\ref{def:mus-bool}, thereby revealing all minimal conflicts in the grounded solution.

\section{Overview of \NAME}
\label{sec:overview}

We present a high-level overview of \NAME, a framework for enumerating minimal unsatisfiable subsets (\texttt{MUSes}) in probabilistic Datalog programs (Algorithm~\ref{alg:mus-enumeration}).
Unlike generic \MUS enumeration methods, \NAME exploits the derivation graph to capture recursive and probabilistic dependencies among facts.
At a high level, it repeatedly finds an \texttt{UNSAT} candidate set $C$ of output facts and shrinks it to a \MUS.

Given a probabilistic Datalog program $\mathbb{P}$, we first invoke a standard solver (e.g., \textsc{ProbLog}~\cite{De2007}) to obtain the queried output facts $\mathcal{O}$ and derivation graph $\mathcal{G}$.
\NAME then encodes $\mathcal{G}$ as Boolean constraints $\Phi$ (Line~\ref{line:encode}) and initializes the search space $\mathsf{Unexp}$ (\emph{Unexplored}) as the power set of $\mathcal{O}$.

\begin{algorithm}[htbp]
{
\footnotesize
\caption{Overview of \NAME  }
\label{alg:mus-enumeration}
\KwIn{
Queried output facts $\mathcal{O}$ and
the derivation graph $\mathcal{G}$
}
\KwOut{All minimal unsatisfiable subsets (\texttt{MUSes})}

\Begin{
    %
    $\mathcal{S} \gets \emptyset$  \textcolor{myblue}{\tcp*{$\mathcal{S}$: the set of MUSes}} \label{line:parse} 
    $\Phi \gets \textsc{Encode}(\mathcal{G})$  \textcolor{myblue}{\tcp*{Generate the encoding of $\mathcal{G}$ (Def~\ref{def:bool-interpret})}} \label{line:encode} 
    $\mathsf{Unexp} \gets \textsc{Powersets}(\mathcal{O},\mathcal{G})$ \label{line:init-unexplored} \\
    $\mathsf{quota} \gets |\mathcal{O}|$ \\
    \While{$\mathsf{Unexp} \neq \emptyset$}{ \label{line:while-unexplored}
        \colorbox{myorange!45}{$\mathcal{C} \gets \textsc{DerivativeSampling}(\mathsf{Unexplored}, \mathsf{quota}, \mathcal{O}, \mathcal{G})$} \label{line:derivative-sampling} 
        \tcp*{\textcolor{myorange!88}{$\bigstar_1:$} Get unexplored cases \textcolor{myorange}{(Sec~\ref{sec:derv_sample})}}
        \If
        (\textcolor{myblue}{\hspace{0.4in} \\
        \tcp*[h]{Check if the combination $\mathcal{C}$ is $\mathsf{UNSAT}$ (Def~\ref{def:joint-sat})}}) 
        {
        $\mathsf{CheckSat}(\mathcal{C}, \Phi) = \mathsf{UNSAT}$}{ \label{line:checkSAT}
            $\mathcal{M} \gets \textsc{Shrink}(\mathcal{C})$
            \label{line:shrink} \\
                $\mathsf{Unexp} \gets \mathsf{Unexp} \setminus (\mathsf{Supersets}(\mathcal{M})\cup \mathsf{Subsets}(\mathcal{M}))$ \\ \label{line:prune-unsat} 
                \colorbox{mypurple!40}{$\mathcal{S} \gets \mathcal{S} \cup \textsc{InferMUS}(\mathcal{M}, \mathsf{Unexp},\mathcal{S})$} \\
                \tcp*[h]{\textcolor{mypurple!88}{$\bigstar_2:$} Bottom-up inference \textcolor{mypurple}{(Sec~\ref{sec:unsatInference})} }\label{line:infer-mus}
                
                $\mathsf{quota} \gets \textsc{ReduceOrMaintainQuota}(\mathsf{quota})$\label{line:reduce-or-maintain-quota} \\
                \tcp*[h]{Half or maintain the quota 
                }\label{line:halfQuota}
        } \Else {
            $\mathsf{Unexp} \gets \mathsf{Unexp} \setminus \mathsf{Subsets}(\mathcal{C})$\; \label{line:prune-sat}
            $\mathsf{quota} \gets \textsc{IncreaseQuota}(\mathsf{quota})$
            \label{line:increase-quota} \textcolor{myblue}{\tcp*{Double the quota 
            }}\label{line:doubleQuota}
            }
        }
        \Return{$\mathcal{S}$}
}
}
\end{algorithm}


\NAME is designed to leverage the structural and statistical information in the derivation graph $\mathcal{G}$ to both prioritize candidate exploration (\textcolor{myorange!88}{$\bigstar_1$}, Line~\ref{line:derivative-sampling}) and aggressively prune the search space (\textcolor{mypurple!88}{$\bigstar_2$}, Line~\ref{line:infer-mus}).
This leads to two main contributions.

First, \NAME identifies \emph{negative statistical dependencies} in the derivation graph (Section~\ref{sec:derv_sample}) and prioritizes combinations of such nodes during candidate generation.
Intuitively, combining negatively dependent derivations is more likely to expose inconsistencies and produce \texttt{UNSAT} candidates.
This idea is realized by \textsc{DerivativeSampling} (Line~\ref{line:derivative-sampling}).

Second, once an MUS $M$ is found, \textsc{InferMUS} (Section~\ref{sec:unsatInference}) exploits the recursive structure of Datalog derivations to infer additional \texttt{UNSAT} candidates in a bottom-up manner.
By propagating conflicts through the derivation graph, \NAME guides the search toward promising regions and prunes likely \texttt{SAT} candidates.

\NAME also includes several mechanisms to improve efficiency.
Since \texttt{MUSes} are often much smaller than the full query set $|\mathcal{O}|$, \NAME dynamically adjusts the sampling quota\footnote{The quota is the number of facts selected in each invocation of \textsc{DerivativeSampling}.} based on satisfiability outcomes (Lines~\ref{line:halfQuota} and~\ref{line:doubleQuota}).
If a sampled candidate $\mathcal{C}$ is satisfiable under $\Phi$ (Line~\ref{line:checkSAT}), the quota is increased, and all subsets of $\mathcal{C}$ are pruned from $\mathsf{Unexp}$.
If a candidate $C$ is unsatisfiable, \NAME shrinks it to a \MUS $M$~\cite{Liffiton2016FastFlexibleMUS}; then, by Definition~\ref{def:mus-bool}, both subsets and supersets of $M$ are pruned from $\mathsf{Unexp}$.

Overall, \NAME tightly integrates \MUS enumeration with probabilistic Datalog semantics by exploiting derivation-graph structure and probabilistic dependencies for both candidate generation and pruning.
Sections~\ref{sec:derv_sample} and~\ref{sec:unsatInference} present these two main ideas.

\section{Derivative Sampling: Prioritizing Search}
\label{sec:derv_sample}

This section presents our first optimization for efficient \MUS generation: \emph{derivative sampling}.
The key idea is to exploit structural and statistical information in the derivation graph to guide exploration toward candidates more likely to be \texttt{UNSAT}, thereby accelerating \MUS discovery.

\subsection{Overview}

\begin{algorithm}[t]
{\footnotesize
\caption{\textsc{DerivativeSampling}}
\label{alg:multi-root-bfs}

\KwIn{Search space of unexplored combinations $\mathsf{U}$}
\KwIn{Derivation graph $\mathcal{G}=(V,E)$, queried output facts $\mathcal{O}$}
\KwIn{Budget $q$ (max \#visited nodes)}
\KwOut{A sampled combination $\mathcal{C}$}

\BlankLine
$\mathsf{Visited} \gets \emptyset$\;
$\mathsf{Roots} \gets$ \textsc{RandSelect}$(\mathcal{O})$\;\label{line:randSelect}
\colorbox{myyellow!45}{$\mathcal{G}' \gets$ \textsc{NegEdgeOnly}($\mathcal{G}$)}\label{line:neg_edge_only}\;\label{line:NegEdgeOnly}
$(V,E) \gets \mathcal{G}'$\;

\Repeat{$\mathsf{Visited}$ unchanged \textbf{or} $|\mathsf{Visited}|\ge q$}{
  $\mathsf{Frontier} \gets \mathsf{Roots} \setminus \mathsf{Visited}$\;
  \While{$\mathsf{Frontier} \neq \emptyset$ \textbf{and} $|\mathsf{Visited}| < q$}{ \label{line:traversal}
    $\mathsf{Visited} \gets \mathsf{Visited} \cup \mathsf{Frontier}$\label{line:union_frontier}\;
    $\mathsf{Next} \gets \{v \mid \exists u \in \mathsf{Frontier}. (u\rightarrow v)\in E \wedge v\notin \mathsf{Visited}\}$\label{line:next_frontier}\;
    $\mathsf{Frontier} \gets \textsc{Take}_{\min(q-|\mathsf{Visited}|,|\mathsf{Next}|)}(\mathsf{Next})$\;
  }

  \If{$|\mathsf{Visited}| < q$}{
    \colorbox{lightgreen}{$S' \gets \textsc{NegNode}(\textsc{LastRoot}(\mathsf{Roots}), \mathsf{Visited}, \mathcal{G})$}\label{line:neg_node}\;
    \If{$S' \notin \mathsf{Visited}$}{
      $\mathsf{Roots} \gets \mathsf{Roots} \cup \{S'\}$\;
    }
  }
}

$\mathsf{Pol} \gets
\{r:\texttt{True} \mid r \in \mathsf{Visited}\}
\cup
\{r:\texttt{False} \mid r \notin \mathsf{Visited}\}$\;

$\mathcal{C} \gets \textsc{Solve}(\mathsf{U},\mathsf{Pol})$\;
\KwRet{$\mathcal{C}$}\;
}
\end{algorithm}

\MUS enumeration is an iterative process of sampling and checking candidate combinations.
Its efficiency depends on the quality of sampled sets; while an \UNSAT combination contains at least one \MUS, a SAT result does not directly advance enumeration.

%


We observe that \texttt{SAT} solving performance degrades as the size of the input formula grows, while \texttt{MUSes} in probabilistic Datalog programs are typically small. This suggests that sampling small combinations with a high likelihood of being \texttt{UNSAT} is essential for efficiency.
Guided by this observation, \emph{Derivative Sampling} (Algorithm~\ref{alg:multi-root-bfs}) exploits the structure of the derivation graph and negative dependencies to achieve three goals: (1) \emph{prioritize negation} by identifying relations that are \underline{\emph{structurally}} likely to conflict (via \colorbox{myyellow!45}{\textsc{NegEdgeOnly}} in Line~\ref{line:NegEdgeOnly}); (2) \emph{limit formula size} by maintaining a small sampling quota~$q$ that bounds each candidate; and (3) \emph{target exploration} by restricting sampling to relations with shared negative (\underline{\emph{statistical}}) dependencies (via \colorbox{lightgreen}{\textsc{NegNode}} in Line~\ref{line:neg_node}).

As shown in Algorithm~\ref{alg:multi-root-bfs}, the procedure constructs a sampled combination $\mathcal{C}$ by strategically traversing the derivation graph $\mathcal{G}$. It begins by initializing the set of visited nodes $\mathsf{Visited}$ and selecting an initial root $\mathsf{Root}$ from the queried output facts $\mathcal{O}$ using \textsc{RandSelect}.

To emphasize relations that are more likely to induce conflicts, the algorithm applies \colorbox{myyellow!45}{\textsc{NegEdgeOnly}} (Algorithm~\ref{alg:neg-edge-only}) to transform $\mathcal{G}$ into a pruned graph $\mathcal{G}'$. This transformation retains only edges corresponding to negative dependencies. Concretely, for a grounded derivation such as
$0.6 :: h \colon\hspace{-0.2em}\mbox{-} \neg b, b_1$,
the original hyperedge $(0.6, h, \allowbreak\{b_1\}, \{b\}, r)$ in $\mathcal{G}$ yields a single negative edge $(h \rightarrow b)$ in $\mathcal{G}'$.

Given the negative-edge-only graph $\mathcal{G}'$, the algorithm performs a BFS within a \texttt{repeat-until} loop under a node budget $q$.
Nodes in the current frontier are added to $\mathsf{Visited}$, and the next frontier is computed by following edges in $\mathcal{G}'$.
If the frontier becomes empty before reaching $q$, the algorithm invokes \textsc{NegNode} (Algorithm~\ref{alg:neg-node}) to select a new node negatively related to the current search path, increasing the chance of generating an \texttt{UNSAT} candidate.

Finally, the algorithm assigns polarities $\mathsf{Pol}$ based on membership in $\mathsf{Visited}$.
Relations in $\mathsf{Visited}$ are more likely to participate in conflicts (\MUS), so their polarity is set to \texttt{true} to prioritize them during sampling, while relations with polarity \texttt{false} are excluded.
These polarity assignments, together with the unexplored search space $\mathsf{U}$, are passed to the solver to generate an unexplored candidate combination $\mathcal{C}$ for satisfiability checking.

\subsection{
Sampling via Negative Dependencies
}


In the overview (Algorithm~\ref{alg:multi-root-bfs}), derivative sampling uses two forms of \emph{negative} dependency—\underline{structural} and \underline{statistical}—to guide candidate selection.
This section describes the corresponding procedures: \colorbox{myyellow!45}{\textsc{NegEdgeOnly}} and \colorbox{lightgreen}{\textsc{NegNode}}.

\begin{algorithm}[h]
\caption{\colorbox{myyellow!45}{\textsc{NegEdgeOnly}}}
\label{alg:neg-edge-only}
{\footnotesize
\KwIn{$\mathcal{G}$}
\KwOut{$\mathcal{G}'$}
$(V,E) \gets \mathcal{G}$\;
$E' \gets \{\, (h \rightarrow b) \mid (\_,h,B^+,B^-,\_) \in E ~\wedge~ b \in B^- \,\}$\;
\KwRet{$(V,E')$}\;
}
\end{algorithm}

\vspace{-0.2in}

\begin{algorithm}[h]
\caption{\colorbox{lightgreen}{\textsc{NegNode}}}
\label{alg:neg-node}
{\footnotesize
\KwIn{$\mathsf{Root}$, $\mathsf{Visited}$, $\mathcal{G}$}
\KwOut{$\mathsf{Root}'$}

\BlankLine
$S^+ \gets
  \{\, r \mid
    \mathsf{CalDep}(r,\mathcal{G},+) \cap \mathsf{CalDep}(\mathsf{Root},\mathcal{G},-) \neq \emptyset
  \,\}$\;
$S^- \gets
  \{\, r \mid
    \mathsf{CalDep}(r,\mathcal{G},-) \cap \mathsf{CalDep}(\mathsf{Root},\mathcal{G},+) \neq \emptyset
  \,\}$\;
$\mathsf{Root}' \gets (S^+ \cup S^-) \setminus \mathsf{Visited}$\;
\KwRet{$\mathsf{Root}'$}\;
}
\end{algorithm}

\paragraph{\textsc{NegEdgeOnly}}
The \textsc{NegEdgeOnly} procedure extracts immediate negative dependencies from the derivation graph $\mathcal{G}$. Consider a hyperedge $(\_, h, B^+, B^-, \_) \in \mathcal{G}$, where $h$ is the head of a grounded rule and $B^+$ and $B^-$ denote its positive and negative body literals, respectively. In the filtered graph, we retain only the flattened negative edges $(h \rightarrow b)$ for each $b \in B^-$. Intuitively, such a negative derivation indicates that the truth of $h$ depends on the falsity of $b$, which increases the likelihood of inconsistency when both are selected. As a result, \textsc{NegEdgeOnly} encourages sampling combinations that include both $h$ and $b$.


\begin{algorithm}[htbp]
\caption{\textsc{CalDep}}
\label{alg:caldep}
{\footnotesize
\KwIn{$r$, $\mathcal{G}$, $\odot \in \{+,-\}$}
\KwOut{$\mathsf{o}$}
$\mathsf{o} \gets \emptyset$\;
\If{$r$ is an input fact}{
  $\mathsf{o} \gets
    \begin{cases}
      \{r\}, & \text{if } \odot = +, \\
      \emptyset, & \text{otherwise}
    \end{cases}$\;
}
\Else{
  \For{$(\_, r, B^+,B^-,\_) \in \mathcal{G}$}{
    \For{$b \in B^+$}{
      $\mathsf{o} \gets
        \begin{cases}
          \mathsf{o} \cup \mathsf{CalDep}(b,\mathcal{G},+), & \text{if } \odot = +, \\
          \mathsf{o} \cup \mathsf{CalDep}(b,\mathcal{G},-), & \text{otherwise}
        \end{cases}$\;
    }
    \For{$b \in B^-$}{
      $\mathsf{o} \gets
        \begin{cases}
          \mathsf{o} \cup \mathsf{CalDep}(b,\mathcal{G},-), & \text{if } \odot = +, \\
          \mathsf{o} \cup \mathsf{CalDep}(b,\mathcal{G},+), & \text{otherwise}
        \end{cases}$\;
    }
  }
}
\KwRet{$\mathsf{o}$}\;
}
\end{algorithm}

\paragraph{\textsc{NegNode}}
While \textsc{NegEdgeOnly} captures local, edge-level negation, many conflicts in probabilistic Datalog arise from \emph{recursive} negative dependencies that span multiple derivation steps. The \textsc{NegNode} procedure addresses this by analyzing negative (\emph{statistical}) dependencies induced by the full derivation structure.

At a high level, two relations can only conflict if they depend on shared input facts in opposite ways. To formalize this intuition, we associate each relation $v$ with two dependency sets:
\begin{itemize}[leftmargin=*,labelsep=0.2em]
    \item $\mathsf{Dep}^+(v)$: input facts supporting $v$ being \emph{true} (positive deps);
    \item $\mathsf{Dep}^-(v)$: input facts supporting $v$ being \emph{false} (negative deps).
\end{itemize}

From a probabilistic perspective, these sets capture statistical influence: if $u \in \mathsf{Dep}^+(v)$, then $\Pr(v \mid u) > \Pr(v)$; if $u \in \mathsf{Dep}^-(v)$, then $\Pr(v \mid u) < \Pr(v)$. Algorithm~\ref{alg:caldep} computes $\mathsf{Dep}^{\odot}(r)$ for a relation $r$ and polarity $\odot \in \{+,-\}$ by recursively traversing the derivation graph. Specifically, for a grounded rule with positive body $B^+$ and negative body $B^-$, the dependencies propagate as follows:
\begin{equation}
\mathsf{Dep}^{\odot}(r)
=
\bigcup_{b \in B^+} \mathsf{Dep}^{\odot}(b)
\;\cup\;
\bigcup_{b \in B^-} \mathsf{Dep}^{\overline{\odot}}(b),
\end{equation}
where $\overline{\odot}$ denotes the opposite polarity. To avoid infinite recursion due to cyclic derivations, we iteratively update $\mathsf{Dep}^+$ and $\mathsf{Dep}^-$ until reaching fix-point.

Using these dependency sets, we define two relations $v_1$ and $v_2$ to be \emph{statistically negatively dependent} if
\[
(\mathsf{Dep}^+(v_1) \cap \mathsf{Dep}^-(v_2) \neq \emptyset)
\;\lor\;
(\mathsf{Dep}^+(v_2) \cap \mathsf{Dep}^-(v_1) \neq \emptyset).
\]

\begin{theorem}
\label{theorem:NegDep}
If two relations $v_1$ and $v_2$ are not \emph{statistically negatively dependent}, then $v_1 \wedge v_2$ is always \texttt{SAT}.
\end{theorem}
\begin{proof}
Since $v_1$ and $v_2$ are not statistically negatively dependent, we have
$
(\mathsf{Dep}^+(v_1) \cap \mathsf{Dep}^-(v_2) = \emptyset)
\;\wedge\;
(\mathsf{Dep}^+(v_2) \cap \mathsf{Dep}^-(v_1) = \emptyset).
$
Thus, no shared input fact appears with opposite polarity in $v_1$ and $v_2$.
If they do not share input facts, assignments for $v_1$ and $v_2$ can be combined, so $v_1 \wedge v_2$ is \texttt{SAT}.
If they do share input facts, the condition above ensures that every shared fact has the same polarity in both relations, so their satisfying assignments can still be made consistent on the shared inputs.
Hence, there exists a joint assignment satisfying both $v_1$ and $v_2$, and therefore $v_1 \wedge v_2$ is \texttt{SAT}.
\end{proof}

Theorem~\ref{theorem:NegDep} motivates \colorbox{lightgreen}{\textsc{NegNode}}.
Given the current root $\mathsf{Root}$, the sampler prioritizes a new root $\mathsf{Root}'$ that is \emph{statistically negatively dependent} on $\mathsf{Root}$ (Algorithm~\ref{alg:neg-node}).
This connects derivation trees with conflicting sources, guiding the search toward combinations more likely to be \texttt{UNSAT} and form \texttt{MUSes}.

\begin{figure}[h]
\centering
\begin{minipage}[t]{0.32\linewidth}
\centering
\includegraphics[width=\linewidth]{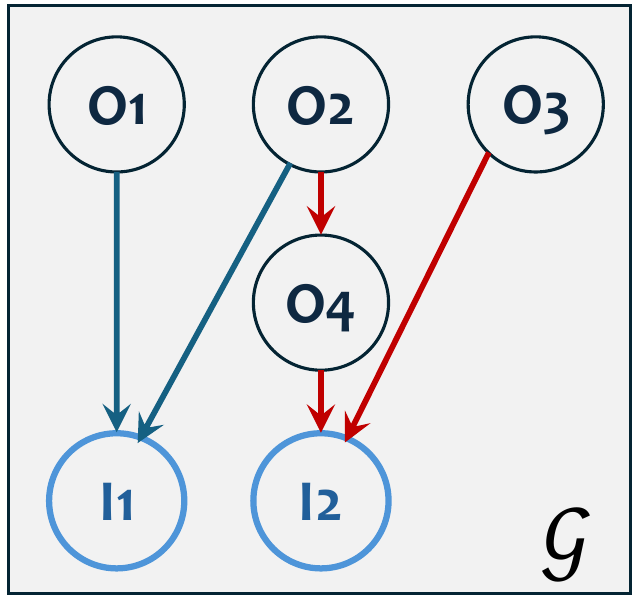}
\end{minipage}\hfill
\begin{minipage}[t]{0.32\linewidth}
\centering
\includegraphics[width=\linewidth]{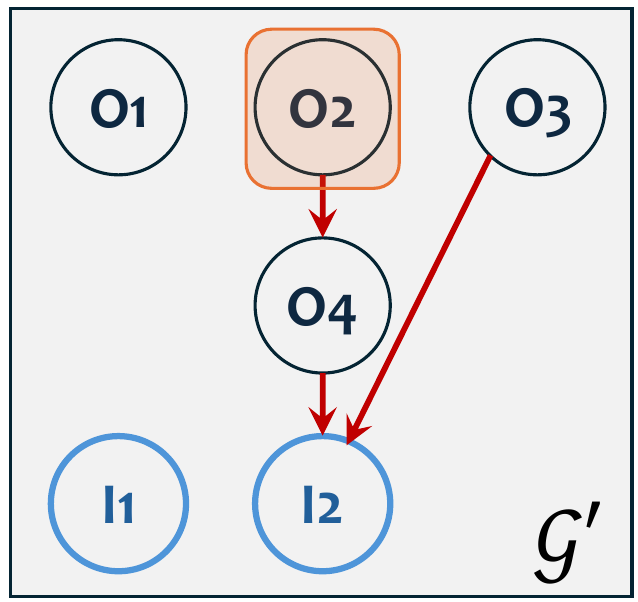}
\end{minipage}\hfill
\begin{minipage}[t]{0.32\linewidth}
\centering
\includegraphics[width=\linewidth]{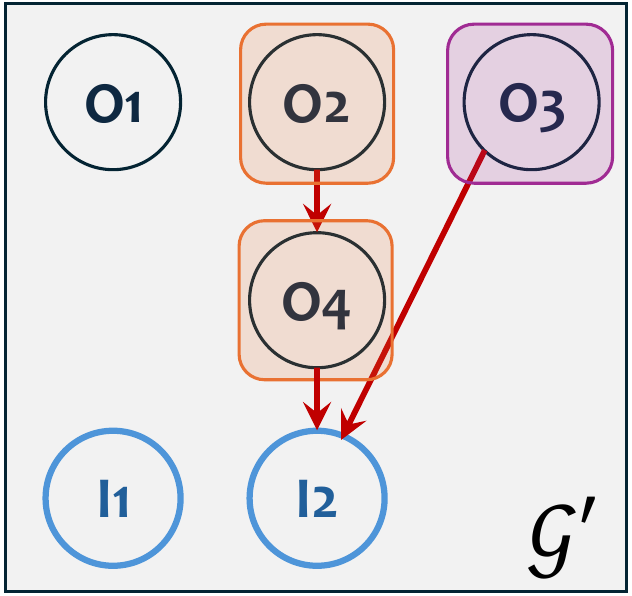}
\end{minipage}
\vspace{-0.1in}
\caption{Sampling example showing the application of \textsc{NegNode}, which selects \textcolor{mypurple}{$\mathsf{O3}$} as the next root from the current root \textcolor{orange}{$\mathsf{O}2$}. The corresponding procedure is shown in Algorithm~\ref{alg:multi-root-bfs}.}
\label{fig:sampling_example}
\vspace{-1em}
\end{figure}

\begin{example}
\label{example:negDep}
Consider the following grounded derivations:
\[
\mathsf{O}1 \colonminus \mathsf{I}1.\quad
\mathsf{O}2 \colonminus \mathsf{I}1, \neg \mathsf{O}4.\quad
\mathsf{O}4 \colonminus \neg \mathsf{I}2.\quad
\mathsf{O}3 \colonminus \neg \mathsf{I}2.
\]
The corresponding derivation graph $\mathcal{G}$ is shown in Figure~\ref{fig:sampling_example}. From the graph structure, $\mathsf{O}2$ inherits the negative dependency of $\mathsf{O}4$. Concretely,
\[
\begin{aligned}
\mathsf{Dep}^+(\mathsf{O}2) &=  \mathsf{Dep}^+(\mathsf{I}1) \cup \underline{\mathsf{Dep}^-(\mathsf{O}4)} \\
&= \mathsf{Dep}^+(\mathsf{I}1) \cup \underline{\mathsf{Dep}^+(\mathsf{I}2)} \\
&= \{\mathsf{I}1,\mathsf{I}2\},
\end{aligned}
\]
and therefore $\mathsf{Dep}^+(\mathsf{O}2) = \{\mathsf{I}1,\mathsf{I}2\}$. Similarly, $\mathsf{Dep}^-(\mathsf{O}3) = \{\mathsf{I}2\}$. Since
\begin{equation}
\label{eq:depO2O3}  
\mathsf{Dep}^+(\mathsf{O}2) \cap \mathsf{Dep}^-(\mathsf{O}3) = \{\mathsf{I}2\},
\end{equation}
$\mathsf{O}2$ and $\mathsf{O}3$ are \emph{statistically negatively dependent}, and sampling them together is likely to expose a conflict.
\end{example}

We illustrate \textsc{DerivativeSampling} (Algorithm~\ref{alg:multi-root-bfs}) on Example~\ref{example:negDep} in Figure~\ref{fig:sampling_example}, using quota $q=3$.

\begin{enumerate}[leftmargin=*]
    \item \textbf{Initial selection (Line~\ref{line:randSelect}).}  
    The sampler begins by randomly selecting a queried relation, for example $\mathsf{O}2$. The visited set is initialized as $\mathsf{Visited}=\{\mathsf{O}2\}$.

    \item \textbf{Negation traversal (Line~\ref{line:traversal}).}  
    The algorithm performs a BFS over negation edges to identify structurally induced conflicts. As indicated by the \textbf{red arrow} in the figure, $\mathsf{O}2$ depends on $\neg\mathsf{O}4$. Consequently, $\mathsf{O}4$ is added to the frontier and included in the visited set, yielding $\mathsf{Visited}=\{\mathsf{O}2, \mathsf{O}4\}$. These relations correspond to the nodes highlighted in the \textbf{orange box}.

    \item \textbf{Conflict-based jump.}
    The BFS frontier becomes empty because $\mathsf{O}4$ has no further negative dependencies, even though the quota is still unmet ($|\mathsf{Visited}|=2 < 3$).
    The algorithm therefore jumps to a new relation outside the current derivation subtree, i.e., beyond the \textbf{orange box}.

    \item \textbf{Extension via negative dependency (Line~\ref{line:neg_node}).}
    The algorithm invokes \textsc{NegNode} to find a relation that is \emph{statistically negatively dependent} on the current visited set.
    Comparing $\mathsf{O}3$ (in the \textbf{purple box}) with $\mathsf{O}2$, we observe a shared dependency on input fact $\mathsf{I}2$ (Eq.~\ref{eq:depO2O3}), indicating a potential conflict.
    Therefore, $\mathsf{O}3$ is added to $\mathsf{Visited}$, yielding $\{\mathsf{O}2, \mathsf{O}4, \mathsf{O}3\}$.

    \item \textbf{Finalization.}
Once the quota is reached, the sampler assigns positive polarities to relations in $\mathsf{Visited}$.
These are passed to the SAT solver as soft constraints to guide generation of an unexplored candidate for satisfiability checking.
\end{enumerate}

\subsection{Dynamic Sampling Quota Control}

This subsection explains the dynamic quota mechanism used in Algorithm~\ref{alg:mus-enumeration} in conjunction with \textsc{DerivativeSampling} (Algorithm~\ref{alg:multi-root-bfs}). Figure~\ref{fig:quota-change} illustrates a representative execution trace, showing how the sampling quota (Y-axis) evolves across successive iterations (X-axis) to balance exploration efficiency against formula complexity.

In the initial iterations (1--3), the algorithm typically starts with a relatively large quota, which is then reduced exponentially. As long as the sampled combinations remain \texttt{UNSAT} (solid markers), the algorithm invokes \textsc{ReduceOrMaintainQuota} (Line~\ref{line:reduce-or-maintain-quota}) to aggressively shrink the sample size. This strategy aims to identify the smallest combination that still preserves unsatisfiability, thereby reducing the cost of subsequent \texttt{SAT} checks.

\begin{figure}[h]
    \centering
    \vspace{-1em}
    \scalebox{0.6}{
    \begin{tikzpicture}[x=1.2cm, y=0.5cm, myblue, thick]
        \draw[->, black] (0.5,0) -- (10,0);
        \draw[->, black] (1,-0.5) -- (1,10);
        
        \node[left, black] at (1, 8) {16};
        \node[left, black] at (1, 4) {8};
        \node[left, black] at (1, 2) {4};
        \node[left, black] at (1, 1) {2};

        \draw[dashed, gray!40, thin] (1,4) -- (8,4); 
        \draw[dashed, gray!40, thin] (1,2) -- (7,2); 
        \draw[dashed, gray!40, thin] (1,1) -- (4,1); 

        \draw[dashed, gray!40, thin] (2,0) -- (2,4);
        \draw[dashed, gray!40, thin] (3,0) -- (3,2);
        \draw[dashed, gray!40, thin] (4,0) -- (4,1);
        \draw[dashed, gray!40, thin] (5,0) -- (5,2);
        \draw[dashed, gray!40, thin] (6,0) -- (6,2);
        \draw[dashed, gray!40, thin] (7,0) -- (7,2);
        \draw[dashed, gray!40, thin] (8,0) -- (8,4);

        \foreach \x in {1, 2, 3, 4, 5, 6, 7, 8}
            \node[below, black] at (\x, -0.3) {\x};

        \begin{scope}[shift={(2.5, 7.2)}]
            \draw[black, fill=white] (0,0) rectangle (2.2, 2.3);
            \filldraw[fill=myblue] (0.4, 1.6) circle (3pt) node[right, black, xshift=3pt] {\small UNSAT};
            \filldraw[fill=white] (0.4, 0.7) circle (3pt) node[right, black, xshift=3pt] {\small SAT};
        \end{scope}

        \draw[myblue, very thick] (1,8) -- (2,4) -- (3,2) -- (4,1) -- (5,2) -- (6,2) -- (7,2) -- (8,4) -- (10,4);

        \foreach \pos in {(1,8), (2,4), (3,2), (5,2), (6,2), (8,4)} {
            \filldraw[fill=myblue] \pos circle (3pt);
            \node[above right, black, font=\small, xshift=-2pt] at \pos {L\ref{line:reduce-or-maintain-quota}};
        }

        \foreach \pos in {(4,1), (7,2)} {
            \filldraw[fill=white] \pos circle (3pt);
            \node[below right, black, font=\small, xshift=-2pt, yshift=-1pt] at \pos {L\ref{line:increase-quota}};
        }

    \end{tikzpicture}
    }
      \vspace{-1em}
    \caption{Evolution of the sampling quota ($y$-axis) across successive iterations ($x$-axis). Solid markers denote UNSAT combinations verified by the SAT solver, whereas hollow markers represent SAT combinations.}
     \vspace{-1em}
    \label{fig:quota-change}
\end{figure}
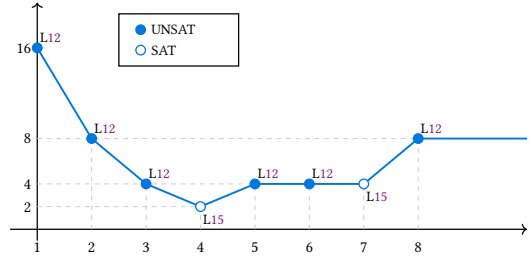

A key transition occurs at iteration~4, where the solver returns a \texttt{SAT} result (hollow marker). This outcome indicates that the current quota (size~2) is too small to expose a conflict. In response, the algorithm triggers \textsc{IncreaseQuota} (Line~\ref{line:increase-quota}) to incrementally enlarge the sampling budget. As a result, the sampler stabilizes around a low but sufficient quota, as observed in iterations~5--7, where the quota remains at size~4 while the solver continues to find \texttt{UNSAT} combinations. By maintaining small candidate sizes whenever possible, the algorithm keeps the induced Boolean formulas compact, which significantly improves \texttt{SAT}-solving performance.

\section{Bottom-Up \texttt{MUS} Inference: Search Pruning}
\label{sec:unsatInference}
To reduce \texttt{SAT} solver calls, we leverage the derivation graph to infer additional \texttt{MUSes} from known conflicts (i.e., previously detected \texttt{MUSes}). Algorithm~\ref{alg:infermus} propagates these conflicts bottom-up through the recursive derivation structure by replacing facts in an existing \MUS with their ancestors, a process we call \emph{logical replacement}.


\begin{algorithm}[htbp]
\caption{$\mathsf{InferMUS}$}
\label{alg:infermus}
{\footnotesize
\KwIn{A known MUS $M$, unexplored combinations $\mathcal{U}$, MUS set $\mathcal{S}$}
\KwOut{Updated MUS set $\mathcal{S}$}

\ForEach{$v \in M$}{
  $\mathcal{N} \gets$ relations exclusively dependent on $v$\;
  \ForEach{$n \in \mathcal{N}$}{
    \colorbox{lightestmain}{$C \gets (M \setminus \{v\}) \cup \{n\}$}\; \label{line-case1}
    $\mathcal{S} \gets \mathsf{MineMUS}(C,\mathcal{U},\mathcal{S})$\;
  }
  $\mathcal{N}' \gets$ relations with two derivations $\{d_1,d_2\}$ where $v \in d_1$\;
  \ForEach{$n' \in \mathcal{N}'$}{
    $\{d_1,d_2\} \gets \mathsf{Derivations}(n')$\;
    \ForEach{literal $l$ in $d_2$}{
      $v' \gets \mathsf{Atom}(l)$\;
      \If{$\mathsf{IsNegative}(l)$}{
        \colorbox{lightestmain}{$C \gets (M \setminus \{v\}) \cup \{n', v'\}$}\; \label{line-case2}
        $\mathcal{S} \gets \mathsf{MineMUS}(C,\mathcal{U},\mathcal{S})$\;
      }
      \Else{
        $\mathcal{S}_{\mathsf{conflict}} \gets \{M_{alt} \in \mathcal{S} \mid v' \in M_{alt}\}$\;
        \ForEach{$M_{alt} \in \mathcal{S}_{\mathsf{conflict}}$}{
          \colorbox{lightestmain}{$C \gets (M \setminus \{v\}) \cup (M_{alt} \setminus \{v'\}) \cup \{n'\}$}\; \label{line-case3}
          $\mathcal{S} \gets \mathsf{MineMUS}(C,\mathcal{U},\mathcal{S})$\;
        }
      }
    }
  }
}
\KwRet{$\mathcal{S}$}\;
}
\end{algorithm}


\begin{algorithm}[htbp]
\caption{$\mathsf{MineMUS}$}
\label{alg:minemus}
{\footnotesize
\KwIn{A candidate $C$, unexplored combinations $\mathcal{U}$, \MUS set $\mathcal{S}$}
\KwOut{Updated MUS set $\mathcal{S}$}

\If{$\mathsf{IsUnexplored}(C)$}{
  $M \gets \textsc{Shrink}(C)$\;
  $\mathcal{U} \gets \mathcal{U} \setminus (\mathsf{Supersets}(M) \cup \mathsf{Subsets}(M))$\;
  $\mathcal{S} \gets \mathcal{S} \cup \{M\} \cup \mathsf{InferMUS}(M,\mathcal{U},\mathcal{S})$\;
}
\KwRet{$\mathcal{S}$}\;
}
\end{algorithm}






\textbf{\emph{Logical Replacement.}}
%
%
The complexity of logical replacement depends on the number of derivation paths available to that ancestor.
An \textbf{Exclusive Dependency} occurs when a relation $n$ depends solely on a relation $v$ already known to be part of an \MUS $M$ (lines 3--6).
In this scenario, $n$ acts as a direct proxy for $v$, and the set $(M \setminus \{v\}) \cup \{n\}$ is guaranteed to be \UNSAT (line~\ref{line-case1}).

Conversely, if a relation $n'$ is \textbf{derived by two derivations} $\{d_1, d_2\}$ where $v \in d_1$, replacing $v$ with $n'$ is insufficient to maintain unsatisfiability because $n'$ could still be satisfied via the alternative path $d_2$.
To maintain the conflict, the algorithm blocks the alternative derivation $d_2$ by targeting its literals at line 10. 
If $d_2$ depends on a negative literal $\neg v'$, the algorithm unions the current set with the atom $v'$ at line~\ref{line-case2}. 
This forces $v'$ to be true, which directly invalidates the alternative path $d_2$. 
%
If $d_2$ depends on a Positive Literal $v'$, the algorithm locates an external blocking set by searching the existing \MUS pool $\mathcal{S}$ for a conflict $M_{alt}$ that contains $v'$ (line 16).
By substituting $v'$ with the remaining members of $M_{alt}$, the algorithm ensures that $v'$ cannot hold, thereby successfully blocking the alternative path $d_2$ at line~\ref{line-case3} to guarantee a new \UNSAT result.

For example, consider MUSes $M = \{O1, O2\}$ and $M_{alt} = \{O3, O4\}$ with rules $O5 \colonminus O1$ and $O5 \colonminus O3$.
%
Replacing $O1$ with $O5$ in $M$ yields the satisfiable set $\{O5, O2\}$ because $O5$ remains true via $O3$ . 
By incorporating $M_{alt}$, the algorithm refines this to the \UNSAT set $\{O5, O4, O2\}$ at line~\ref{line-case3}, where $O4$ acts as a blocking agent for the $O3$ path.
%
%

\textbf{\emph{The Mining Procedure.}}
This replacement logic is orchestrated by the \texttt{MineMUS} (Algorithm~\ref{alg:minemus}), which serves as the primary engine for conflict discovery.
When provided with an \UNSAT candidate $C$, it first applies \textsc{Shrink}\cite{Liffiton2016FastFlexibleMUS}--a standard deletion-based minimization procedure--to yield a minimal \MUS $M$ satisfying  Definition~\ref{def:mus-bool}.
%
%
%
%
After identifying a new \MUS $M$, the algorithm prunes the search space $\mathcal{U}$ by removing all supersets of $M$ (which are non-minimal) and all subsets of $M$ (which are guaranteed to be \texttt{SAT}) at line 4.
Finally, it invokes \texttt{InferMUS} to expand the conflict through the derivation graph, generating new candidates for mining without additional solver overhead.

\section{Evaluation}
\label{sec:eval}

We implemented our approach in a prototype tool called \NAME, developed in C++. The tool extends the Soufflé Datalog engine~\cite{jordan2016souffle} by instrumenting its resolution procedure to extract a complete derivation graph for the input probabilistic Datalog program. 
For satisfiability checking, \NAME integrates with the CVC5 SMT solver~\cite{barbosa2022cvc5}. 

Our experimental evaluation was conducted on a Linux-based system powered by an AMD EPYC 7R13 processor (32 physical cores, 64 logical threads) and 242 GB of RAM, operating within a KVM-based virtualization environment. 
We assigned a 30-minute time limit (\texttt{T/O}) to each benchmark execution. The evaluation is designed to address the following research questions:

\begin{itemize}[leftmargin=*]
    \item \textbf{RQ1:} How efficient and scalable is our method in generating \texttt{MUSes} compared to state-of-the-art baseline approaches?
    \item \textbf{RQ2:} What is the quality of the \texttt{MUSes} generated by our method, in terms of both the total number of identified conflicts and the resulting coverage of Datalog relations?
    \item \textbf{RQ3:} How helpful are the generated \texttt{MUSes} compared to the baselines in assisting real-world Datalog-based program analysis tasks?
\end{itemize}




\subsection{Applications, Benchmarks, and Baselines}

We evaluate \NAME on \textbf{70} benchmarks drawn from \textbf{four} distinct domains, as summarized in Table~\ref{tab:benchmark-stat}.

\begin{table}[h]
\centering
\caption{Benchmark Statistics}
\label{tab:benchmark-stat}
\vspace{-1em}
\small
\setlength{\tabcolsep}{4pt}
\scalebox{0.99}{
\begin{tabular}{l|cc|cc|cc}
& \multicolumn{2}{c|}{\textbf{Average}} & \multicolumn{2}{c|}{\textbf{Max}} & \multicolumn{2}{c}{\textbf{Median}} \\
& \#nodes & \#edges & \#nodes & \#edges & \#nodes & \#edges \\
\hline
\textbf{SC}~\cite{wang2025OOPSLA, zhang2018scinfer}     & 29,798 & 71,267 & 166,530 & 389,545 & 4,294 & 14,210 \\
\textbf{DR}~\cite{Raghothaman2018,li2025combining}        & 59,829 & 254,279 & 119,398 & 722,116 & 59,874 & 210,046 \\
\textbf{SD}~\cite{Sung2018} & 8,251 & 638,252 & 50,849 & 5,884,061 & 2,819 & 56,339 \\
\textbf{Bayes}~\cite{bench-munin}            & 4,309 & 55,343 & 4,461 & 66,188 & 4,400 & 55,823 \\
\end{tabular}
}
\vspace{-0.1in}
\end{table}

\begin{figure*}[t]
  \centering
  \begin{subfigure}[t]{0.24\textwidth}
    \centering
    \includegraphics[width=\linewidth]{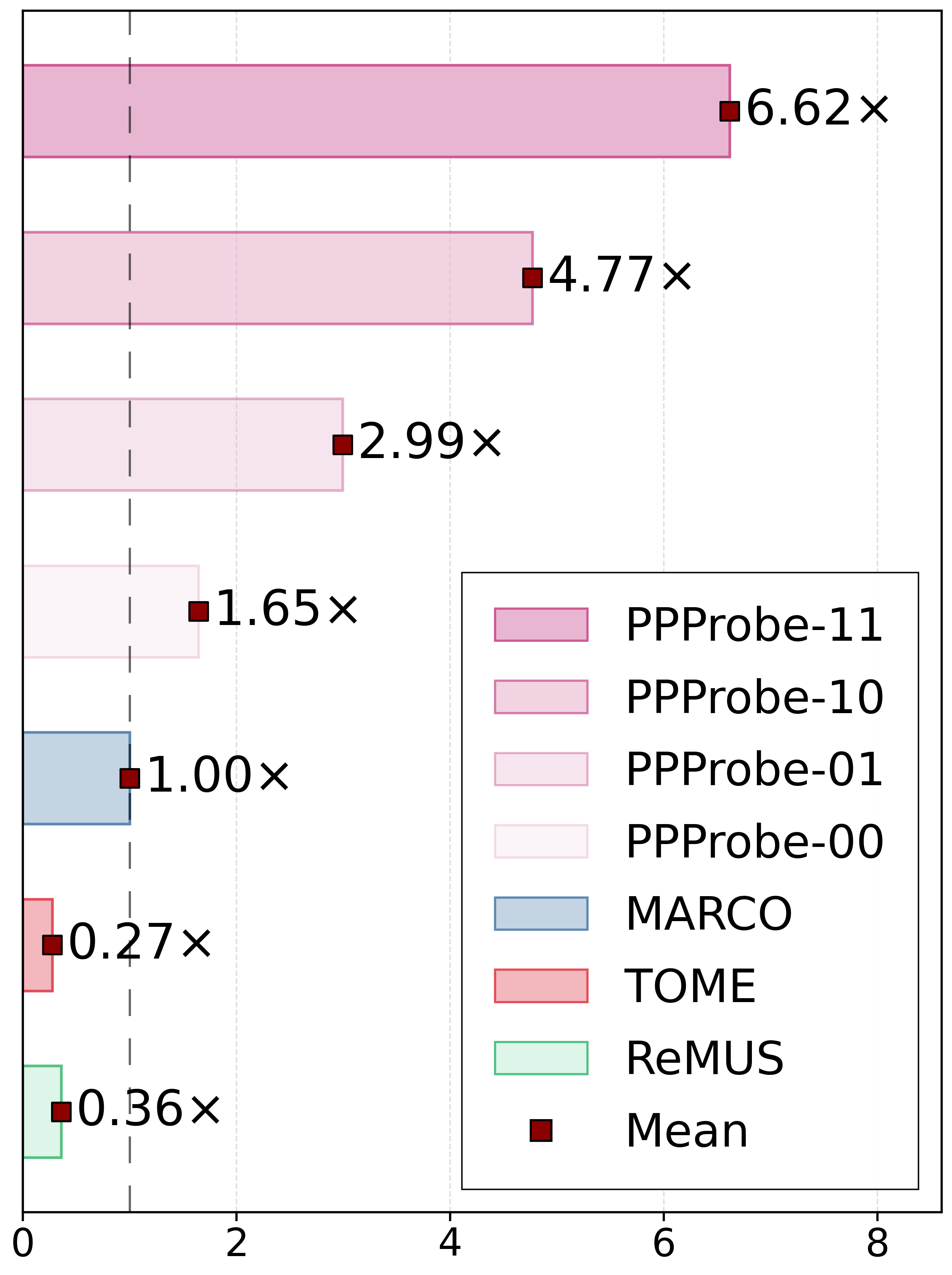}
    \caption{Side-channel Analysis (SC)}
    \label{fig:figure6-side-chans}
  \end{subfigure}\hfill
  \begin{subfigure}[t]{0.24\textwidth}
    \centering
    \includegraphics[width=\linewidth]{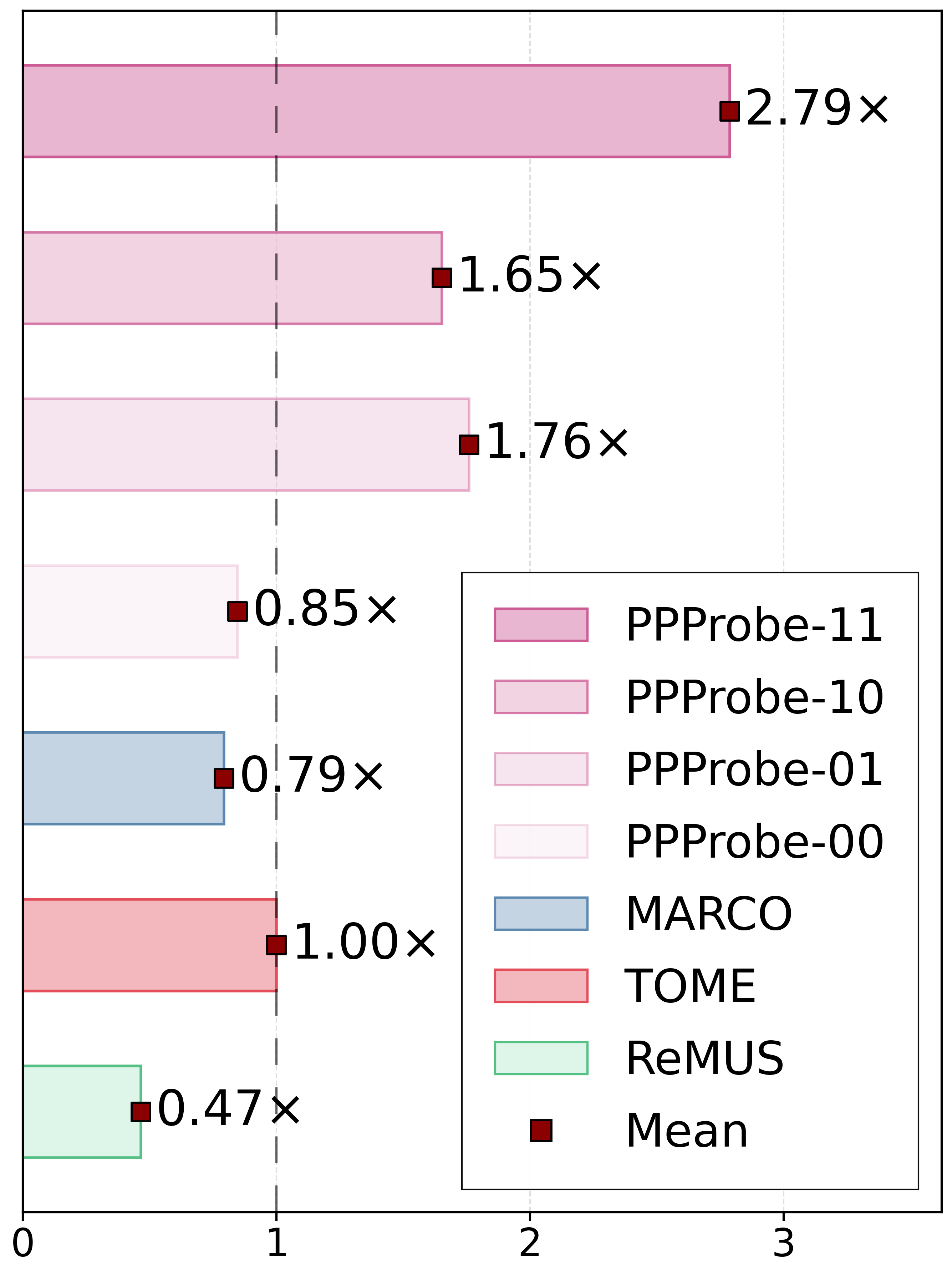}
    \caption{Semantic Diffing (SD)}
    \label{fig:figure6-semantic-diffing}
  \end{subfigure}\hfill
  \begin{subfigure}[t]{0.24\textwidth}
    \centering
    \includegraphics[width=\linewidth]{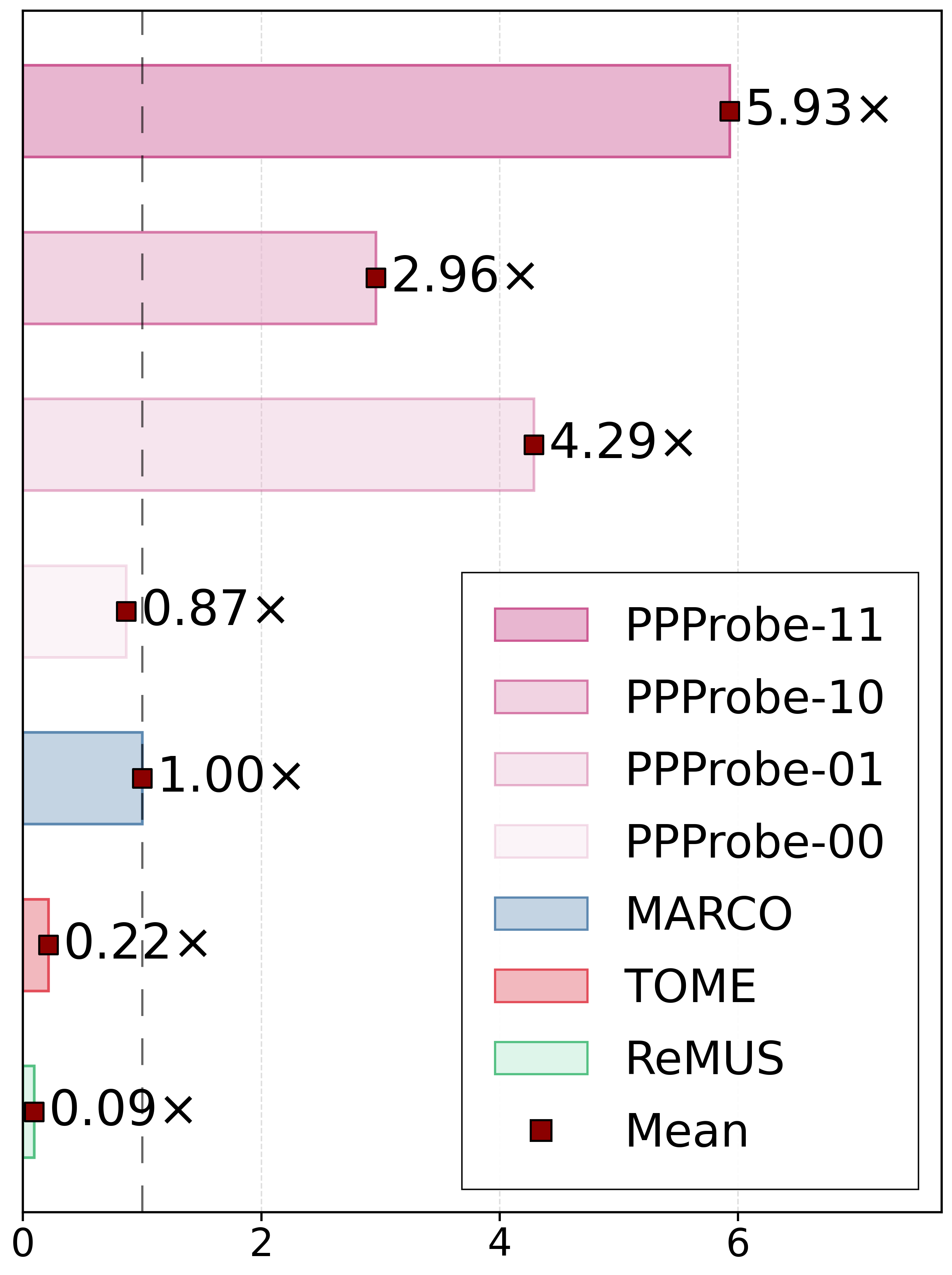}
    \caption{Data Race Analysis (DR)}
    \label{fig:figure6-datarace}
  \end{subfigure}\hfill
  \begin{subfigure}[t]{0.24\textwidth}
    \centering
    \includegraphics[width=\linewidth]{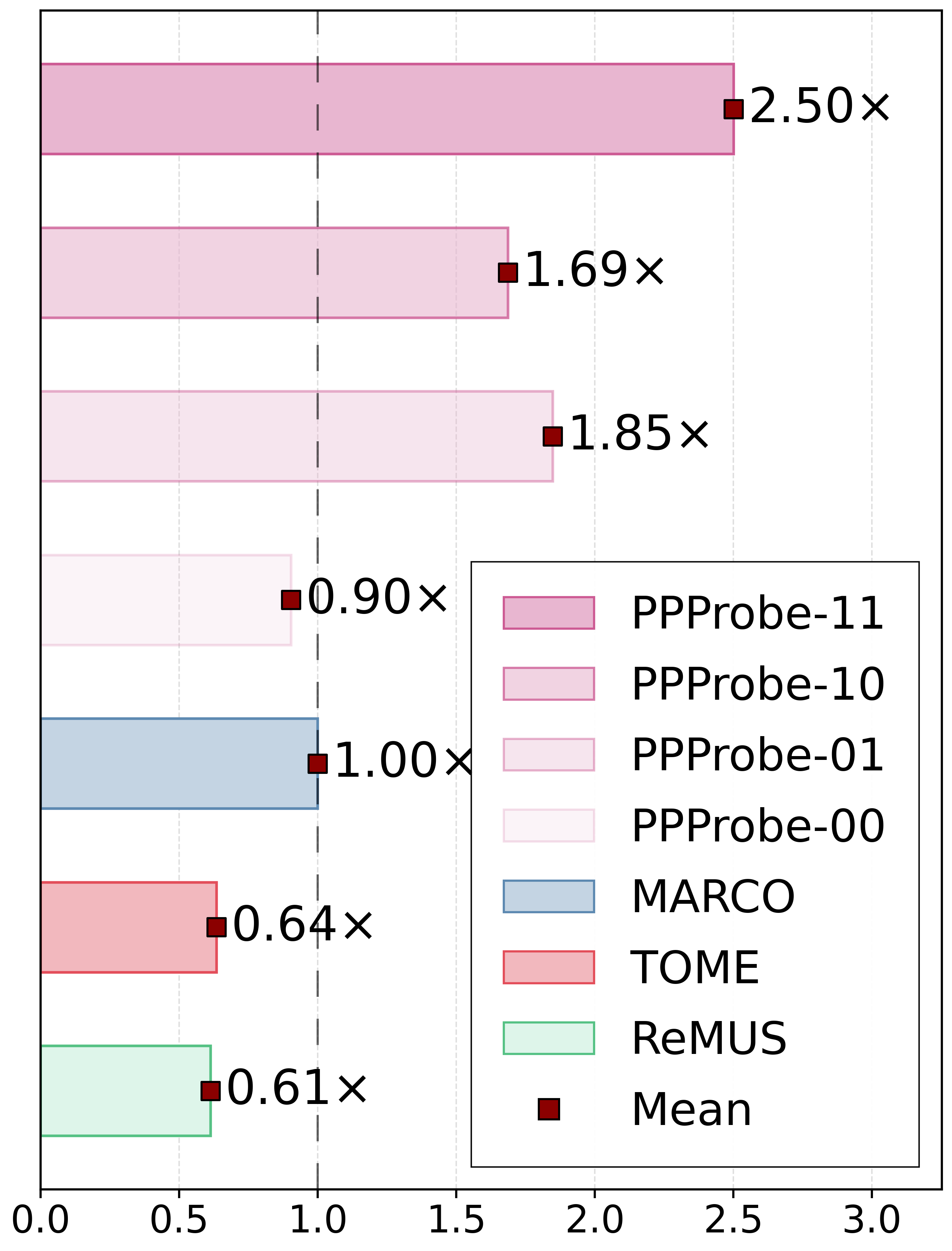}
    \caption{Bayes}
    \label{fig:figure6-bnr}
  \end{subfigure}
  \vspace{-0.1in}
  \caption{Normalized \MUS count of \NAME and its ablated variants relative to the strongest baseline. 
  }
  \label{fig:figure-four-domains}
\end{figure*}

\textbf{Side-Channel Analysis.} The first domain (labeled \textbf{SC} in Table~\ref{tab:benchmark-stat}) consists of \textbf{18} Datalog-based analyses for detecting power side-channel leaks~\cite{wang2025OOPSLA,wang2021data,zhang2018scinfer}. We include this domain because side-channel risk assessment requires estimating the likelihood of information leakage rather than checking a purely Boolean property. Recent work~\cite{wang2025OOPSLA} has further introduced a probabilistic Datalog formulation for modeling such likelihoods. In this evaluation, the analyses are applied to implementations of widely used cryptographic primitives, including AES, SHA-3, and MAC-Keccak, with queries of the form \texttt{key\_sensitive} to identify potential secret leaks for further investigation.

\textbf{Data Race Analysis.} The second domain (labeled \textbf{DR} in Table~\ref{tab:benchmark-stat}) consists of \textbf{7} data-race analyses on real-world concurrent applications such as the Apache FTP server, the ETH web crawler and AVR microcontroller simulator.
Prior works~\cite{Raghothaman2018,Zhang2017} developed probabilistic Datalog rules to detect data races in those applications that combine thread-escape, may-happen-in-parallel, and lockset analyses in a flow- and context-sensitive manner.

In this setting, the input facts are deterministic, while the Datalog rules are probabilistic, capturing the fraction of cases in which a rule may produce an incorrect conclusion from true hypotheses. The queries are alarms of the form \texttt{race(p1,p2)}, indicating that program points \texttt{p1} and \texttt{p2} may participate in a data race and require further inspection.

\textbf{Semantic Diffing.} The third domain (labeled \textbf{SD} in Table~\ref{tab:benchmark-stat}) consists of \textbf{41} semantic-diffing analyses for identifying synchronization differences between evolving multithreaded programs including the Apache HTTP server, the FreeBSD audit subsystem, and the Linux IIO driver. 
Each application compares two program versions and looks for synchronization, ordering, or data-flow behavior present in one version but not the other.
Following prior work~\cite{Sung2018}, these analyses are formulated as Datalog programs that compute differentiating data-flow edges, namely edges permitted by one program but not the other. In this setting, the input facts are deterministic, while the Datalog rules are probabilistic, capturing the uncertainty introduced by approximate inference. The queries correspond to synchronization differences that may indicate unintended behavioral changes and require further inspection.


\textbf{Bayesian Networks (Bayes).} The fourth domain includes \textbf{4} large-scale Bayesian networks from the \texttt{bnlearn} repository~\cite{scutari2010learning}. 
We focus on \texttt{Massive Networks} (over 1,000 nodes) to test scalability. 
We convert these networks into probabilistic Datalog by designating source nodes (those without incoming edges) as input facts and all other nodes as output facts. 
%

\textbf{Benchmark Statistics.} Table~\ref{tab:benchmark-stat} summarizes key statistics for the four benchmark categories. Here, \#nodes and \#edges denote the numbers of nodes and hyperedges in the corresponding derivation graphs. For each category, we report the average, maximum, and median values across all benchmarks. 
These statistics indicate the problem scale indirectly. 
In our encoding, nodes and hyperedges induce fact and rule variables respectively. 
And probabilistic inference reduces to weighted model counting (WMC), which is \#P-hard and worst-case exponential in the number of Boolean random variables determined by \#nodes and \#edges.

%

\textbf{Baselines and Experimental Setup.} We compare \NAME against \MARCO~\cite{Liffiton2016FastFlexibleMUS}, \REMUS~\cite{Bendik2018ReMUS}, and \TOME~\cite{Bendik2016TOME}, which are state-of-the-art domain-agnostic \MUS enumeration methods. These baselines are designed for propositional formulas in conjunctive normal form and do not natively support the hierarchical and recursive dependencies in Datalog. For a fair comparison, we integrated them into the \texttt{MUST} framework~\cite{Bendik2020MUST} and extended them with our sound encoding (Sec.~\ref{sec:conflict-encoding}). As a result, all baselines share the same grounding and formula construction phases and can process nested conjunctions and disjunctions in probabilistic Datalog.

\subsection{Efficiency and Scalability (RQ1)}


To evaluate the efficiency of \NAME, we measured its runtime on all benchmarks under the timeout limit and compared it against three state-of-the-art baselines: \MARCO~\cite{Liffiton2016FastFlexibleMUS}, \REMUS~\cite{Bendik2018ReMUS}, and \TOME~\cite{Bendik2016TOME}. We also include three ablation variants of \NAME to isolate the contributions of its core components:
\begin{itemize}[leftmargin=*, labelsep=0.2em]
\item \textbf{\NAME-00}: \NAME without either optimization.
\item \textbf{\NAME-10}: \NAME with derivative sampling only (Sec.~\ref{sec:derv_sample}).
\item \textbf{\NAME-01}: \NAME with bottom-up inference only (Sec.~\ref{sec:unsatInference}).
\item \textbf{\NAME-11}: Full \NAME with both optimizations.
\end{itemize}


\begin{table}[h]
\centering
\footnotesize
\setlength{\tabcolsep}{4pt}
\vspace{-0.1in}
\caption{Runtime Statistics for Completed Benchmarks 
}
\scalebox{1}{ 
\begin{tabular}{l rrrr rrrr rrrr}
\toprule
 &  \multicolumn{4}{c}{\textbf{Average }} &  \multicolumn{4}{c}{\textbf{Median }} &  \multicolumn{4}{c}{\textbf{Maximum }} \\
    \cmidrule(lr){2-5}  \cmidrule(lr){6-9} \cmidrule(lr){10-13}
 & \textbf{P}  & \textbf{M} & \textbf{R} & \textbf{T}  & \textbf{P} & \textbf{M} & \textbf{R} & \textbf{T}  & \textbf{P} & \textbf{M} & \textbf{R} & \textbf{T} \\
\midrule
\textbf{SC}  & \textbf{0.4} & 0.5 & 0.4 & 2.4 & \textbf{0.0} & 0.0 & 0.0 & 0.0 & \textbf{1.9} & 3.0 & 2.4 & 17.9 \\
\bottomrule
\end{tabular}
}
\label{tab:time-finished}
\vspace{-0.1in}
\end{table}

\textbf{\emph{Terminated Instances.}} Eight \textbf{side-channel analysis} benchmarks completed within the time limit. For these instances, Table~\ref{tab:time-finished} reports the runtime (in seconds) of the full \NAME method (\NAME-11) and three baseline MUS enumerators: \MARCO, \REMUS, and \TOME, abbreviated as \textbf{M}, \textbf{R}, and \textbf{T}, respectively. The table summarizes the \textbf{average}, \textbf{median}, and \textbf{maximum} runtime across all eight benchmarks.
We report only the full \NAME-11 configuration, since the four ablation variants exhibited nearly identical performance on these terminated instances.
The median runtimes of \NAME-11 and all three baselines (\textbf{M}, \textbf{R}, and \textbf{T}) round to zero because the values are very small.
Overall, \NAME-11 remains competitive with the baseline \MUS enumerators: it achieves the best average and maximum runtime among all methods. 
%


\textbf{\emph{Timed-Out Instances.}}
For the remaining benchmarks that did not terminate within the \texttt{T/O} limit, we evaluate efficiency by counting the number of \texttt{MUSes} discovered within the 30-minute budget.
Figure~\ref{fig:figure-four-domains} reports, for each \NAME variant and each baseline (\MARCO, \TOME, and \REMUS), the normalized throughput relative to the \emph{best} baseline in the corresponding application domain; that is, we divide the number of \texttt{MUSes} found by a given method by the number found by the strongest baseline. 
%
Because the strongest baseline differs across domains, the normalization is domain-specific.
As shown in Figure~\ref{fig:figure-four-domains}, \TOME is the best baseline for \textbf{Semantic Diffing}, whereas \MARCO is the best baseline for the other three domains.

More specifically, Figure~\ref{fig:figure-four-domains}(a) reports the average normalized throughput over the 18 \textbf{side-channel analysis} benchmarks.
On average, \NAME-11 discovers $6.62\times$ more \texttt{MUSes} than \MARCO, the strongest baseline in this domain, within the same \texttt{T/O} budget.
Its gains over \TOME and \REMUS are even larger: \NAME-11 finds about $24\times$ more \texttt{MUSes} than \TOME and about $18\times$ more than \REMUS.
A similar trend also appears in the other domains shown in Figure~\ref{fig:figure-four-domains}(b)--(d). 

\textbf{\emph{Ablation.}} Regarding the ablation variants, both \NAME-10 and \NAME-01 substantially outperform the unoptimized \NAME-00, indicating that each optimization contributes meaningfully to performance on its own.
Furthermore, their combination in \NAME-11 leads to a much larger gain.
%

\textbf{Summary.} Overall, these results show that \NAME-11 consistently achieves the highest throughput, especially on large search spaces, confirming that derivative sampling and \MUS inference are complementary and mutually reinforcing.

\begin{table*}[h]
    \centering
    \small
    \caption{Aggregated reduction statistics (in \%).}
    \label{tab:reduction-statistics}
    \resizebox{0.99\columnwidth}{!}{
    \begin{tabular}{lc rrrr rrrr rrrr}
      \toprule 
      & \multirow{2}{*}{\textbf{Query}}
      & \multicolumn{4}{c}{\textbf{Average $\Delta$\%}}
      & \multicolumn{4}{c}{\textbf{Median $\Delta$\%}}
      & \multicolumn{4}{c}{\textbf{Maximum $\Delta$\%}} \\
      \cmidrule(lr){3-6}
      \cmidrule(lr){7-10}
      \cmidrule(lr){11-14}
      & & \textbf{P} & \textbf{M} & \textbf{R} & \textbf{T}
      & \textbf{P} & \textbf{M} & \textbf{R} & \textbf{T}
      & \textbf{P} & \textbf{M} & \textbf{R} & \textbf{T} \\
      \midrule
      \textbf{SC} & \texttt{key\_sensitive(var)}
      & \textbf{69.17} & 54.94 & 38.53 & 51.02
      & \textbf{87.96} & 65.11 & 18.27 & 47.92
      & \textbf{100.00} & \textbf{100.00} & \textbf{100.00} & \textbf{100.00} \\
      \midrule

      \textbf{DR} & \begin{tabular}[c]{@{}c@{}}\texttt{ultimateRace(p1,p2)}\\\texttt{racePairs(p1,p2)}\end{tabular}
      & \textbf{60.77} & 13.64 & 0.12 & 1.37
      & \textbf{79.25} & 1.93 & 0.00 & 0.40
      & \textbf{99.86} & 59.91 & 0.72 & 6.42 \\
      \midrule

      \textbf{SD} & \texttt{mayHb(s1,s2,p)}
      & \textbf{13.16} & 12.19 & 5.22 & 8.60
      & \textbf{12.98} & 11.60 & 2.94 & 8.90
      & 23.08 & \textbf{23.35} & 20.44 & 20.44 \\
      \bottomrule
    \end{tabular}
    }
  \end{table*}

\subsection{Quality of Generated MUSes (RQ2)}

While the efficiency of an \MUS generator is often evaluated by throughput, the utility of the generated conflicts also depends on their diagnostic diversity.
In probabilistic Datalog analysis, discovering a large number of \texttt{MUSes} is useful only if they uncover conflicts spanning a broader range of Datalog relations or facts across different parts of the program logic.
To quantify this aspect, we use \emph{relation coverage} ($\#\mathsf{Rel}$), defined as the number of distinct Datalog relations covered by the \texttt{MUSes} discovered within the \texttt{T/O}.

\begin{table}[htbp]
\centering
\footnotesize
\setlength{\tabcolsep}{4pt}
\caption{Quality of \texttt{MUSes}: Relation Coverage Count}
\label{tab:mus-quality}
\vspace{-0.1in}
\scalebox{1.0}{
\begin{tabular}{l rrrr rrrr rrrr}
\toprule
 & \multicolumn{4}{c}{\textbf{Average}} 
 & \multicolumn{4}{c}{\textbf{Median}} 
 & \multicolumn{4}{c}{\textbf{Maximum}} \\
\cmidrule(lr){2-5} \cmidrule(lr){6-9} \cmidrule(lr){10-13}
 & \textbf{P} & \textbf{M} & \textbf{R} & \textbf{T}
 & \textbf{P} & \textbf{M} & \textbf{R} & \textbf{T}
 & \textbf{P} & \textbf{M} & \textbf{R} & \textbf{T} \\
\midrule
\textbf{SC}    & \textbf{4012} & 1571 & 247 & 1011 & \textbf{3923} & 1549 & 219 & 612 & \textbf{7692} & 3795 & 441 & 2224 \\
\textbf{DR}    & \textbf{6791} & 1509 & 127 & 308  & \textbf{6942} & 1505 & 149 & 381 & \textbf{9710} & 2688 & 174 & 453 \\
\textbf{SD}    & \textbf{1157} & 711  & 171 & 333  & \textbf{641}  & 571  & 174 & 307 & \textbf{3467} & 2046 & 320 & 705 \\
\textbf{Bayes} & \textbf{1333} & 1054 & 143 & 537  & \textbf{1284} & 1026 & 153 & 499 & \textbf{1588} & 1367 & 173 & 715 \\
\bottomrule
\end{tabular}
}
\end{table}

Table~\ref{tab:mus-quality} reports $\#\mathsf{Rel}(X)$ for each method $X$. Here, $X$ ranges over four methods: \NAME, \MARCO, \REMUS, and \TOME, abbreviated as \textbf{P}, \textbf{M}, \textbf{R}, and \textbf{T}, respectively.
Each row corresponds to one application domain (\textbf{SC}, \textbf{DR}, \textbf{SD}, and \textbf{Bayes}), and each column summarizes the aggregated results for one method.
For each method and domain, we report the average, median, and maximum values of $\#\mathsf{Rel}$ across all benchmarks.
For example, the entry in the \textbf{\textcolor{lightmain}{SC}} row and the \textbf{\textcolor{darkermain}{M}} column under \textbf{Average} denotes the average value of $\#\mathsf{Rel}(\mathsf{\textcolor{darkermain}{\textbf{MARCO}}})$ over all \textcolor{lightmain}{\textbf{side-channel}} benchmarks.

\textbf{\emph{Diagnostic Diversity through Coverage.}} In program analysis tasks such as \textbf{SC}, \textbf{DR}, and \textbf{SD}, higher relation coverage suggests that the discovered \texttt{MUSes} capture a broader range of code statements, variables, and information-flow paths that contribute to side-channel vulnerabilities (\textbf{SC}), data races (\textbf{DR}), and unintended synchronization changes (\textbf{SD}). As shown in Table~\ref{tab:mus-quality}, on side-channel analysis (\textbf{SC}) benchmarks, \NAME covers, on average, about $16\times$ more relations than \REMUS (i.e., $4012$ vs.\ $247$), $4\times$ more than \TOME, and $2.5\times$ more than \MARCO. Similar trends can also be seen in other program-analysis domains, including \textbf{DR} and \textbf{SD}.

\textbf{Summary.}
These results suggest that baseline tools such as \REMUS often revisit similar logical paths, whereas \NAME, guided by structure-aware optimizations such as bottom-up inference, is better able to explore distinct branches of the derivation graph.
As a result, within the same \emph{time budget}, \NAME covers a larger portion of the program relations (\emph{more diverse} and \emph{diagnostically useful}) and gives developers a more complete view of potential inconsistencies in probabilistic Datalog programs.

\subsection{Case Study: Reducing False Alarms (RQ3)}

To address RQ3, we evaluate the practical utility of \NAME and the baselines in three real-world static analysis settings: reducing false positives in side-channel analysis (\textbf{SC}), data race analysis (\textbf{DR}), and semantic diffing (\textbf{SD}) for evolving multithreaded programs.
In \textbf{SC}, the key queries are \texttt{key\_sensitive(var)}, which indicate potential secret leakage through program variable \texttt{var}.
In \textbf{DR}, the key queries are \texttt{ultimateRace(p1,p2)} and \texttt{racePairs(p1,p2)}, which indicate potential data races between program points \texttt{p1} and \texttt{p2}.
In \textbf{SD}, the key queries are \texttt{mayHb(s1,s2,p)}, which denote may-happen-before edges in program \texttt{p} and capture execution orders possible under some thread interleavings.
Since these analyses are sound but incomplete, they often produce many spurious reports, increasing manual debugging effort and unnecessary mitigation costs.
We therefore study whether the MUSes generated by \NAME can act as an automated filter by identifying and removing logically inconsistent reasoning paths.

\textbf{Methodology.}
We reconstruct scenarios in which developers use the \texttt{MUSes} generated by \NAME and baselines  (\texttt{MARCO}, \texttt{ReMUS} and \texttt{TOME}) to complement domain knowledge and prune the diagnostic search space that developers must inspect manually.
Each \MUS is encoded as a hard constraint, together with a small number of developer-provided observations on non-critical queries.
The critical queries $S$ (e.g., \texttt{key\_sensitive} in \textbf{SC}) are encoded as soft constraints.
A MaxSAT solver then computes a largest consistent subset $S' \subseteq S$, and we measure the reduction as $\Delta\% = (|S| - |S'|)/|S|$.
Without \MUS constraints, all critical queries would remain potentially consistent and no reduction would occur; the \texttt{MUSes} provide the logical conflicts needed for pruning.
Because multiple solutions may satisfy the hard constraints, MaxSAT returns one that preserves the largest number of alarms, so $\Delta\%$ gives a lower bound on the achievable false-positive reduction.
A higher reduction rate allows developers to focus on alarms more likely to represent genuine correctness or security issues, lowering the cost of manual triage, unnecessary fixes, and mitigation.


\textbf{Comparison Results.} 
Table~\ref{tab:reduction-statistics} summarizes the aggregated reduction rates (\textbf{AVG}, \textbf{MED}, and \textbf{MAX}) achieved by \NAME (P), \texttt{MARCO}(M), \texttt{ReMUS}(R) and \texttt{TOME}(T) for each application domain across all benchmarks.
Both \textbf{SC} and \textbf{DR} show substantial opportunities for false-positive reduction, with average reductions of $69.17\%$ and $60.77\%$ achieved by \NAME, which are both higher than the baselines.
Notably, for the \emph{KS\_transition} benchmark in \textbf{SC}~\cite{bench-P9-18}, \MUS-based pruning removes all reported leaks (\emph{max:} $100\%$), indicating that all detected alarms are false positives.
We manually inspected these cases and confirmed that the removed \texttt{key\_sensitive} queries are benign.
A similar result appears in the \emph{weblech} benchmark from \textbf{DR}, where the reduction reaches $99.86\%$.

\textbf{SD} achieves a lower average reduction rate ($13.16\%$) by \NAME, mainly because its \texttt{MUSes} cover fewer relevant relations.
This is consistent with Table~\ref{tab:mus-quality}, which shows that \textbf{SD} attains lower relation coverage than \textbf{SC} and \textbf{DR}.
In the 41 benchmarks of \textbf{SD} domain, \NAME  achieves the highest false-alarm reduction rate in 35 benchmarks and \texttt{MARCO} wins the rest 6 benchmarks, which explains the reason why \NAME has a slightly lower maximum reduction rate than \texttt{MARCO} ($23.08\%$ vs. $23.35\%$) but has a higher average ($13.16\%$ vs. $12.19\%$) and median reduction rate ($12.98\%$ vs. $11.60\%$).

\textbf{Summary.} 
Overall, these results show that the \texttt{MUSes} generated by \NAME can substantially shrink the diagnostic search space and provide a lower bound on the achievable false-positive reduction, and outperforms the baselines on most of the benchmarks.

\subsection{Stability Analysis}

To assess the randomness of \NAME brought by derivative sampling in Algorithm~\ref{alg:multi-root-bfs}, we repeat each benchmark \textbf{three} times with distinct random seeds and compute the relative standard deviation (RSD) of the number of MUSes discovered within the fixed timeout.

\begin{wraptable}{l}{0.43\columnwidth}
\hspace{-1in}
 \vspace{-0.1in}
    \centering
    \small
    \caption{RSD of \#\texttt{MUSes}.
    }
    \vspace{-0.1in}
    \label{tab:ppprobe-stability}
    \setlength{\tabcolsep}{2.3pt}
      \begin{tabular}{lrrr}
        \toprule
        \textbf{Domain}
          & \textbf{Avg}
          & \textbf{Median}
          & \textbf{Max} \\
        \midrule
        \textbf{SC}    & 7.59\% & 7.51\% & 17.32\% \\
        \textbf{DR}    & 3.66\% & 3.37\% &  7.41\% \\
        \textbf{SD}    & 7.64\% & 5.25\% & 46.00\% \\
        \textbf{Bayes} & 6.28\% & 6.51\% &  9.93\% \\
        \bottomrule
      \end{tabular}%
    \vspace{-0.1in}
\end{wraptable}

Shown as Table~\ref{tab:ppprobe-stability}, \NAME exhibits generally stable performance across seeds: the median RSD ranges from only $3.37\%$ to $7.51\%$ across domains, while the average RSD remains below $8\%$ in all cases.
The $46.00\%$ RSD is an isolated outlier, as the next two highest RSDs drop substantially to $25.24\%$ and $22.03\%$.
It's caused by a single instance, where different seeds trigger inference cascades of substantially different sizes before the timeout. 
Overall, the low domain-level median and average RSDs demonstrate that \NAME's performance is stable across random seeds.

\section{Related Work}
\label{sec:related}

\noindent
\textit{\textbf{Datalog for Program Analysis.}}
Datalog is one of the popular techniques in program analysis such as detecting data race ~\cite{Naik2006, Raghothaman2018, Zhang2017}, detecting side-channel~\cite{wang2019mitigating, wang2021data}, points-to-analysis~\cite{Bravenboer2009, Madsen2016, Smaragdakis2014, Whaley2005, Zhang2014}, program diffing~\cite{Sung2018} and thread-modular analysis~\cite{Kusano2016, kusano2017}.
%
%
Wang et al.~\cite{wang2019mitigating} analyzed power side channels using Datalog, focusing on how key information can leak during register allocation in compilation.  
However, their approach cannot quantify the likelihood of such leaks, limiting its practicality for real-world scenarios.
Zhang et al.~\cite{Zhang2014} and Raghothaman et al.~\cite{Raghothaman2018} incorporated probability into Datalog to address this limitation.  
However, their approaches overlook potential conflicts in probabilistic Datalog programs, which can restrict the expressiveness of the analysis and lead to redundant or false alarms.
Our approach identifies these conflicts, significantly reducing false positives and allowing developers to focus on genuine vulnerabilities.
%

\noindent
\textit{\textbf{Probabilistic Datalog.}}
A wide range of probabilistic logic programming languages extend Prolog or Datalog with quantitative reasoning, including PRISM~\cite{Taisuke1995PRISM}, LPADs~\cite{Vennekens2004LPAD}, Blog~\cite{Milch2007BLOG}, CP-logic~\cite{Vennekens2009CPLOGIC}, PPDL~\cite{Grohe2022PPDL}, Datalogp~\cite{Fuhr1995DATALOGP}, ProbLog~\cite{De2007}, and Scallop~\cite{li2023scallop}.
Most of these systems reduce probabilistic inference to weighted model counting (WMC) and employ symbolic representations such as BDDs for efficiency.
%
%
%
While these languages advance the expressiveness and efficiency of probabilistic inference, they overlook conflicting probabilistic outputs.
In contrast, our work introduces the first framework for conflict detection in probabilistic Datalog with minimal unsatisfiable subsets.

\noindent
\textit{\textbf{Minimal Unsatisfiable Subsets (MUSes).}}
MUS enumeration has been studied in many domains, including linear programming~\cite{Van1981, Chinneck1991, Gleeson1990}, propositional satisfiability~\cite{Dershowitz2006, Lynce2004, Mneimneh2005, Oh2004}, temporal logic~\cite{Bendik2017LTL}, and numerical CSPs~\cite{Gasca2007}. Close to our setting, OCE~\cite{Shlyakhter2003} and RCE~\cite{Torlak2008} extract unsatisfiable cores from declarative specifications but assume reducibility to propositional logic, whereas our domain involves recursive Datalog derivations.
Domain-agnostic algorithms instantiate black-box SAT/SMT oracles and fall into two families: (1)~hitting-set–based methods such as CAMUS~\cite{LiffitonSakallah2005, Liffiton2008}, DAA~\cite{Bailey2005DAA}, and PDDS~\cite{Stern2012PDD}, which exploit the MUS--MCS duality; and (2)~seed-shrink methods such as MARCO~\cite{Liffiton2016FastFlexibleMUS}, TOME~\cite{Bendik2016TOME}, ReMUS~\cite{Bendik2018ReMUS}, and UNIMUS~\cite{Bendik2020Unimus}, which iteratively find and minimize unsatisfiable subsets.
Among these, MARCO, TOME, and ReMUS are applicable to probabilistic Datalog in principle but ignore derivation-graph structure, scaling poorly on such encodings as our experiments confirm.
\section{Conclusion}
\label{sec:conclusion}
In this work, we addressed the critical but overlooked challenge of conflicting outputs in probabilistic Datalog, which can undermine the reliability of quantitative program analysis. 
We introduced \NAME, the first framework to soundly and efficiently identify these conflicts by representing them as minimal unsatisfiable subsets (\texttt{MUSes}). 
\NAME leverages structural and statistical dependencies within Datalog derivation graphs to prioritize and prune search space by derivative sampling and bottom-up \UNSAT inference.
Our extensive evaluation on 70 real-world benchmarks demonstrated that \NAME outperforms state-of-the-art baselines.
Applied as a post-analysis filter, the extracted MUSes reduce the diagnostic search space by an average of 69\% and 61\% on side-channel and data-race benchmarks, respectively, providing a conservative lower bound on false-positive elimination.

\section{Data Availability Statement}
\NAME, along with the baselines and benchmarks, is available via \href{https://zenodo.org/records/21776435}{\textcolor{blue!60}{the artifact link}}.


\bibliographystyle{ACM-Reference-Format}
\bibliography{reference}

\end{document}